%% file: Full_paper.tex
\documentclass{article}
\usepackage[nonatbib,final]{nips_2016}

\usepackage{setspace} 

\usepackage[utf8]{inputenc} 
\usepackage[T1]{fontenc}    
\usepackage{url}            
\usepackage{booktabs}       
\usepackage{amsfonts}       
\usepackage{nicefrac}       
\usepackage{microtype}      
\usepackage{times}
\usepackage{tikz}
\usetikzlibrary{patterns}
\usepackage{subcaption}
\usepackage{version}
\usepackage{amssymb}
\usepackage{mathtools}
\usepackage{algpseudocode}
\usepackage{mathrsfs,graphicx,algorithm,color}

\newtheorem{thm}{Theorem}
\newtheorem{assumption}{Assumption}
\newtheorem{lemma}{Lemma}
\newenvironment{proof}{\paragraph{Proof:}}{\hfill$\square$}

\DeclareMathOperator*{\argmin}{\arg\!\min}
\DeclareMathOperator*{\argmax}{\arg\!\max}
\newcommand{\qed}{\ensuremath{\square}}
\newcommand{\diag}{{\mathbf{diag}}}
\newcommand{\tr}{{\mathbf{tr}}}

\newcommand{\dom}{{\mathbf{dom}\;}}
\providecommand{\inner}[2]{\left\langle {#1} , {#2} \right\rangle}		   
\providecommand{\norm}[1]{\left\|#1 \right\|}
\providecommand{\dnorm}[1]{\left\|#1 \right\|^{*}}

\providecommand{\norml}[1]{\left\|#1 \right\|_{1}}

\title{Worst Case Competitive Analysis of Online Algorithms for Conic Optimization}

 \author{
 Reza Eghbali \\
  Department of Electrical Engineering\\
  University of Washington\\
  Seattle, WA 98195 \\
  \texttt{eghbali@uw.edu} \\
   \And
    Maryam Fazel \\
  Department of Electrical Engineering\\
  University of Washington\\
  Seattle, WA 98195 \\
  \texttt{mfazel@uw.edu} \\
}

  \includeversion{full}
\excludeversion{nips}
  
\begin{document}

\maketitle

\begin{abstract}%
  \input FullAbstract.tex

\end{abstract}

\section{Introduction}

Given a proper convex cone $K \subset \mathbb{R}^n$, let $\psi: K \mapsto \mathbb{R}$ be an upper semi-continuous concave function. Consider the optimization problem
\begin{align}\label{main}
 \mbox{maximize}&{  \quad \psi\left(\sum_{t=1}^{m} A_t x_t\right)}\\ \notag
\mbox{subject to}& \quad  x_t \in F_t ,  \quad  \forall t \in [m],
\end{align}
where for all $t\in[m] := \{1,2,\ldots,m\}$, $x_t\in\mathbb{R}^k$ are the optimization variables and $F_t$ are compact convex constraint sets. We assume $A_t \in \mathbb{R}^{n\times k}$ maps $F_t$ to $K$; for example, when $K=\mathbb{R}^n_+$ and $F_t \subset \mathbb{R}^k_+$, this assumption is satisfied if $A_t$ has nonnegative entries.
We consider problem \eqref{main} in the online setting, where it 
can be viewed as a sequential game between a player (online algorithm) and an adversary.  At each step $t$, the adversary reveals $A_t,\;F_t$ and the algorithm chooses $\hat{x}_t \in F_t$. The performance of the algorithm is measured by its competitive ratio, i.e., the ratio of objective value at $\hat{x}_1, \ldots, \hat{x}_m$ to the offline optimum. 

Problem \eqref{main} covers (convex relaxations of) various online combinatorial problems including online bipartite matching \cite{karp1990optimal}, the ``adwords'' problem \cite{mehta2007Adwords}, and the secretary problem \cite{kleinberg2005multiple}.  More generally, it covers online linear programming (LP) \cite{buchbinder2009online}, online packing/covering with convex cost \cite{azar2014online,buchbinder2014online,chan2015online}, and generalization of adwords \cite{devanur2012online}. Online LP is an important example on the positive orthant:
\begin{align*}
 \mbox{maximize}&{  \quad \sum_{t=1}^{m} c_t^T x_t}\\ \notag
\mbox{subject to}& \quad \sum_{t=1}^{m} B_t x_t \leq b\\ \notag
& \quad  \mathbf{1}^T x_t \leq 1, \; x_t \in \mathbf{R}_{+}^k,    \quad  \forall t \in [m].
\end{align*}
Here $F_t = \{x \in \mathbf{R}_{+}^k ~|~ \mathbf{1}^T x \leq 1\}$, $A_t^T = [c_t , B_t^T]$ and $\psi(u,v) = u + I_{\{\cdot \leq b\}}(v)$ where $I_{\{\cdot \leq b\}}(v)$ is the concave indicator function of the set $\{v \leq b\}$.

The competitive performance of online algorithms has been studied mainly under the worst-case model (e.g., in \cite{mehta2007Adwords}) or stochastic models (e.g., in \cite{kleinberg2005multiple}). In the worst-case model one is interested in lower bounds on the competitive ratio that hold for any $(A_1, F_1), \ldots,(A_m , F_m)$. 
In stochastic models, adversary chooses a probability distribution from a family of distributions to generate $(A_1, F_1),\ldots,(A_m, F_m)$, and the competitive ratio is calculated using the expected value of the algorithm's objective value. The two most studied stochastic models are random permutation and i.i.d. models. In the random permutation model the adversary is limited to choose distributions that are invariant under random permutation while in the i.i.d. model $(A_1, F_1),\ldots,(A_m, F_m)$ are required to be independent and identically distributed. Note that the worst case model can also be viewed as an stochastic model if one allows distributions with singleton support.

Online bipartite matching, and its generalization the ``adwords'' problem, are two main examples that have been studied under the worst case model. The greedy algorithm achieves a competitive ratio of ${1}/{2}$ while the optimal algorithm achieves a competitive ratio of $1-1/e$ as bid to budget ratio goes to zero \cite{mehta2007Adwords,buchbinder2007online,karp1990optimal,kalyanasundaram2000optimal}. 
A more general version of adwords in which each agent (advertiser) has a concave cost has been studied in \cite{devanur2012online}. On the other hand, secretary problem and its generalizations are stated for random permutation model since the worst case competitive ratio of any online algorithm can be arbitrary small \cite{babaioff2008online,kleinberg2005multiple}. The convex relaxation of the online secretary problem is a simple linear program (LP) with one constraint. Therefore, the competitive ratio of online algorithms for LP has either been analyzed under random permutation or i.i.d models \cite{agrawal2009dynamic,feldman2010online,devanur2011near,jaillet2012near,kesselheimRTV13} or under the worst case model with further restrictions on the problem data \cite{buchbinder2009online}.
  Several authors have also proposed algorithms for adwords and bipartite matching problems under stochastic models which have a better competitive ratio than the competitive ratio under the worst case model \cite{mahdian2011online,karande2011online,feldman2009online2,manshadi2012online,jaillet2013online,devanur2009Adwords}.

The majority of algorithms proposed for the problems mentioned above rely on a primal-dual framework \cite{buchbinder2007online,buchbinder2009online,azar2014online,devanur2012online,buchbinder2014online}.
The differentiating point among the algorithms is the method of updating the dual variable at each step, since once the dual variable is updated the primal variable can be assigned using a simple assignment rule based on complementary slackness condition. A simple and efficient method of updating the dual variable is through a first order online learning step. For example, the algorithm stated in \cite{devanur2011near} for online linear programming uses mirror descent with entropy regularization (multiplicative weight updates algorithm) once written in the primal dual language. Recently, the work in \cite{devanur2011near} was independently extended to the random permutation model in \cite{gupta2014experts,eghbali2014exponentiated,agrawal2014fast}. In \cite{agrawal2014fast}, the authors provide competitive difference bound for online convex optimization under random permutation model as a function of the regret bound for the online learning algorithm applied to the dual. 

In this paper, we consider two versions of the greedy algorithm for problem \eqref{main} 
The first algorithm, Algorithm \ref{algorithm1}, updates the primal and dual variables sequentially. The second algorithm, Algorithm \ref{algorithm2}, provides a 
direct saddle-point representation of what has been described 
informally in the literature as ``continuous updates'' of primal and dual variables.  This saddle point representation allows us to generalize this type of updates to non-smooth function.
In section \ref{CA}, we bound the competitive ratios of the two algorithms. A sufficient condition on the objective function that guarantees a non-trivial worst case competitive ratio is introduced. We show that the competitive ratio is at least $\frac{1}{2}$ for a monotone non-decreasing objective function. 
Examples that satisfy the sufficient condition (on the positive orthant and the positive semidefinite cone) are given. 
In section \ref{s::smoothing}, we derive optimal algorithms, as variants of greedy algorithm applied to a smoothed version of $\psi$. For example, Nesterov smoothing provides optimal algorithm for online LP. The main contribution of this paper is to show how one can derive the optimal smoothing function (or from the dual point of view the optimal regularization function) for separable $\psi$ on positive orthant by solving a convex optimization problem. This gives a implementable algorithm that achieves the optimal competitive ratio derived in \cite{devanur2012online}. We also show how this convex optimization can be modified for the design of the smoothing function specifically for the sequential algorithm. In contrast, \cite{devanur2012online} only considers continuous updates algorithm.

  
{\bf Notation.} We denote the transpose of a matrix $A$ by $A^T$. The inner product on $\mathbb{R}^n$ is denoted by $\inner{\cdot}{\cdot}$. 
Given a function $\psi: \mathbb{R}^n \mapsto \mathbb{R}$, $\psi^*$ denotes the concave conjugate of $\psi$ defined as 
\begin{align*}
\psi^*(y) = \inf_{u} \inner{y}{u} - \psi(u),
\end{align*}
for all $y \in \mathbb{R}^n$. For a concave function $\psi$, $\partial \psi(u)$ denotes the set of supergradients of $\psi$ at $u$, i.e., the set of all $y \in \mathbb{R}^n$ such that
\begin{align*}
\forall u' \in \mathbb{R}^n: \quad \psi(u') \leq \inner{y}{u' - u} + \psi(u).
\end{align*}
The set $\partial \psi$ is related to the concave conjugate function $\psi^*$ as follows. For an upper semi-continuous concave function $\psi$ we have 
\begin{align*}
\partial \psi(u) = \argmin_y \inner{y}{u} - \psi^*(y).
\end{align*}
Under this condition, $(\psi^*)^* = \psi$. 

A differentiable function $\psi$ has a Lipschitz continuous gradient with respect to $\norm{\cdot}$ with continuity parameter $\frac{1}{\mu} > 0$ if for all $u,u' \in \mathbb{R}^n$, 
\[\dnorm{\nabla \psi(u') - \nabla \psi(u)} \leq \frac{1}{\mu} \norm{u-u'},\]
where $\dnorm{\cdot}$ is the dual norm to $\norm{\cdot}$. For an upper semi-continuous concave function $\psi$, this is equivalent to $\psi^*$ being $\mu$-strongly concave with respect to $\dnorm{\cdot}$ (see, for example, \cite[chapter 12]{rockafellar1998variational}). 

The dual cone $K^*$ of a cone $K \subset \mathbb{R}^n$ is defined as $K^* = \{y \;\,|\;\, \inner{y}{u} \geq 0 \;\, \forall u \in K \}$. 
Two examples of self-dual cones are the positive orthant $\mathbb{R}_{+}^n$ and the cone of $n\times n$ positive semidefinite matrices $S_{+}^n$. A proper cone (pointed convex cone with nonempty interior) $K$ induces a partial ordering on $\mathbb{R}^n$ which is denoted by $\leq_{K}$ and is defined as \begin{nips}$x \leq_{K} y   \Leftrightarrow y-x \in K.$\end{nips}
\begin{align*}
x \leq_{K} y   \Leftrightarrow y-x \in K.
\end{align*}
For two sets $F,G \subset  \mathbb{R}^n$, we write $F \leq_K G$ when $u \leq_K v$ for all $u \in F, v \in G$.


\subsection{Two greedy algorithms}\label{two-algorithms}

The (Fenchel) dual problem for problem \eqref{main} is given by
\begin{equation}\label{fenchel-dual}
 \mbox{minimize} {\quad \sum_{t=1}^{m} \sigma_{t}(A_t^T y) - \psi^{*}(y)},
\end{equation}
where the optimization variable is $y\in\mathbb{R}^n$, and $\sigma_t$ denotes the \emph{support function} for the set $F_t$ defined as $\sigma_{t}(z) = \sup_{x \in F_{t}} \inner{x}{z}$. We denote the optimal dual objective with $D^\star$.

A pair $(x^*,y^*) \in (F_1 \times \ldots F_m) \times K^*$ is an optimal primal-dual pair if and only if
\begin{align} \label{primal-optimality}
 x^*_t &\in \argmax_{x \in F_t} \inner{x}{A_t^{T} y^*}\quad \forall t \in [m],\\ \label{dual-optimality}
 y^* &\in \partial \psi(\sum_{t=1}^{m} A_t x_t^*).
\end{align}
Based on these optimality conditions, we consider two algorithms. Algorithm \ref{algorithm1} updates the primal and dual variables \emph{sequentially}, by maintaining a dual variable $\hat{y}_t$ and using it to assign $\hat{x}_t \in \argmax_{x \in F_t} \inner{x}{A_t^{T} \hat{y}_t}$. The algorithm then updates the dual variable based on the second optimality condition \eqref{dual-optimality} \footnote{we assume that $\partial \psi$ nonempty on the boundary of $K$, otherwise, we can modify the algorithm to start from a point in the interior of the cone very close to zero.}. By the assignment rule, we have $A_t \hat{x}_t \in \partial \sigma_{t}(\hat{y}_t)$, and the dual variable update can be viewed as
\begin{align*}\hat{y}_{t+1} \in \argmin_{y} \inner{\sum_{s=1}^{t} A_s \hat{x}_s}{y} - \psi^*(y).\end{align*}
Therefore, the dual update is the same as the update in dual averaging \cite{nesterov2009primal} or 
Follow The Regularized Leader (FTRL) \cite{shalev2007primal,abernethy2008competing}\cite[section 2.3]{shalev2011online} algorithm with regularization $-\psi^*(y)$. 
\begin{algorithm}
\caption{Sequential Update}
   \label{algorithm1}
{\fontsize{10}{15}\selectfont
    \begin{algorithmic}
    \State {Initialize $\hat{y}_1 \in \partial \psi(0)$}
      \For{$t \gets 1 \textrm{ to } m$}
      \State{Receive $A_t, F_t$}
       \State{$ \hat{x}_t \in \argmax_{x \in F_{t}} \inner{x}{A_t^{T} \hat{y}_t}$}
       \State{$\hat{y}_{t+1} \in \partial \psi(\sum_{s=1}^{t} A_s \hat{x}_s)$}
      \EndFor
       \end{algorithmic}
       }
\end{algorithm}

Algorithm \ref{algorithm2} updates the primal and dual variables \emph{simultaneously}, ensuring that
 \[
 \tilde{x}_t \in \argmax_{x \in F_t} \inner{x}{A_t^{T} \tilde{y}_t}, \qquad \tilde{y}_{t} \in \partial \psi( \sum_{s=1}^{t} A_s \tilde{x}_s).
\]
This algorithm is inherently more complicated than algorithm \ref{algorithm1}, since finding $\tilde{x}_t$ involves solving a saddle-point problem. This can be solved by a first order method like the mirror descent algorithm for saddle point problems. In contrast, the primal and dual updates in algorithm \ref{algorithm1} solve two separate maximization and minimization problems \footnote{Also if the original problem is a convex relaxation of an integer program, meaning that each $F_t = {\rm conv} \mathcal{F}_t$ where $\mathcal{F}_t \subset \mathbb{Z}^l$, then $\hat{x}_t$ can always be chosen to be integral while integrality may not hold for the solution of the second algorithm.} 
\begin{algorithm}
\caption{Simultaneous Update}
   \label{algorithm2}
{\fontsize{10}{15}\selectfont
    \begin{algorithmic}
      \For{$t \gets 1 \textrm{ to } m$}
      \State{Receive $A_t, F_t$}
\State{$(\tilde{y}_t, \tilde{x}_t ) \in \arg\min_{y }\max_{x \in F_t} \;\inner{y}{A_t x + \sum_{s=1}^{t-1} A_s \tilde{x}_s} - \psi^*(y) $ }   
      \EndFor
       \end{algorithmic}
       }
\end{algorithm}

For a reader more accustomed to optimization algorithms, we would like to point to alternative views on these algorithms. If $0 \in F_t$ for all $t$, then an alternative view on Algorithm \ref{algorithm2} is coordinate minimization. Initially, all the coordinates are set to zero, then at step $t$ 
\begin{align}\label{coordinate}
\tilde{x}_t \in \argmax_{x \in F_t} \psi\left(\sum_{s=1}^{t} A_t \tilde{x}_s + A_t x\right).
\end{align}  

If in addition $\psi$ is differentiable, Algorithm \ref{algorithm1}, at each time step $t$, applies one iteration of Frank-Wolfe algorithm \cite{frank1956algorithm} for solving \eqref{coordinate}.
\section{Competitive ratio bounds and examples for $\psi$}\label{CA} 
In this section, we derive bounds on the competitive ratios of Algorithms \ref{algorithm1} and \ref{algorithm2} by bounding their respective duality gaps. Let $\tilde{y}_{m+1}$ to be the minimum element in $\partial \psi(\sum_{t=1}^{m} A_t \tilde{x}_t)$ with respect to ordering $\leq_{K^*}$ (such an element exists in the superdifferential by Assumption \eqref{antitone}, which appears later in this section). We also choose $\hat{y}_{m+1}$ to be the minimum element in $\partial \psi(\sum_{t=1}^{m} A_t \hat{x}_t)$ with respect to $\leq_{K^*}$. Note that $\hat{y}_{m+1}$ is used for analysis and does not play a role in the assignments of Algorithm \ref{algorithm1}. 
Let $P_{\rm seq}$ and $P_{\rm sim}$ denote the primal objective values for the primal solution produced by the algorithms \ref{algorithm1} and \ref{algorithm2}, and $D_{\rm seq}$ and $D_{\rm sim}$ denote the corresponding dual objective values,
\begin{align*}
P_{\rm seq} &=\psi\left(\sum_{t=1}^{m} A_t \hat{x}_t\right), \; &&P_{\rm sim} = \psi\left(\sum_{t=1}^{m} A_t \tilde{x}_t\right), \\
D_{\rm seq} &= \sum_{t=1}^{m} \sigma_t(A_t^T \hat{y}_t) - \psi^{*}(\hat{y}_{m+1}), \; &&D_{\rm sim} = \sum_{t=1}^{m} \sigma_t(A_t^T \tilde{y}_t) - \psi^{*}(\tilde{y}_{m+1}).
\end{align*}

The next lemma provides a lower bound on the duality gaps of both algorithms.
\begin{nips}
\begin{lemma}\label{duality-gap} The duality gaps for the two algorithms can be lower bounded as 
\[
P_{\rm sim} - D_{\rm sim} \geq \psi^*(\tilde{y}_{m+1}) + \psi(0), \quad
P_{\rm seq} - D_{\rm seq} \geq \psi^*(\hat{y}_{m+1}) +\psi(0) + \sum_{t=1}^{m} \inner{ A_t \hat{x}_t}{ \hat{y}_{t+1} - \hat{y}_t}
\]
Furthermore, if $\psi$ has a Lipschitz continuous gradient with parameter ${1}/{\mu}$ with respect to  $\norm{\cdot}$, 
\begin{align}\label{regret}
P_{\rm seq} - D_{\rm seq} \geq \psi^*(\hat{y}_{m+1}) +\psi(0) - \frac{1}{2 \mu}\sum_{t=1}^{m} \norm{ A_t \hat{x}_t}^2.
\end{align}
\end{lemma}
\end{nips}
\begin{full}
\begin{lemma}\label{duality-gap} The duality gap for Algorithm \ref{algorithm2} is lower bounded as
\begin{align*}P_{\rm sim} - D_{\rm sim} \geq \psi^*(\tilde{y}_{m+1}) + \psi(0),\end{align*}
and the duality gap for Algorithm \ref{algorithm1} is lower bounded as
\begin{align*}P_{\rm seq} - D_{\rm seq} \geq \psi^*(\hat{y}_{m+1}) +\psi(0) + \sum_{t=1}^{m} \inner{ A_t \hat{x}_t}{ \hat{y}_{t+1} - \hat{y}_t}. \end{align*} 
Furthermore, if $\psi$ has a Lipschitz continuous gradient with parameter ${1}/{\mu}$ with respect to $\norm{\cdot}$, then
\begin{align}\label{regret}
P_{\rm seq} - D_{\rm seq} \geq \psi^*(\hat{y}_{m+1}) +\psi(0) - \sum_{t=1}^{m} \frac{1}{2 \mu}\norm{ A_t \hat{x}_t}^2.
\end{align}
\end{lemma}
\end{full}
Note that right hand side of \eqref{regret} is exactly the regret bound of the FTRL algorithm (with a negative sign) \cite{shalev2007online}. The proof 
is given in Appendix \ref{proofs}.  
To simplify the notation in the rest of the paper, we assume $\psi(0)= 0$.

Now we state a sufficient condition on $\psi$ that leads to non-trivial competitive ratios, and we assume this condition holds in the rest of the paper. 
One can interpret this assumption as having ``diminishing returns'' with respect to the ordering induced by a cone. Examples of functions that satisfy this assumption will appear later in this section.

\begin{assumption}\label{antitone}
Whenever $u \geq_K v$, there exists $y \in \partial \psi(u)$ that satisfies \begin{nips}$y \leq_{K^*} z $\end{nips}
\begin{full}
\begin{align*}
y \leq_{K^*} z, 
\end{align*}
\end{full}
for all $z \in \partial \psi(v)$.
\end{assumption}
When $\psi$ is differentiable, Assumption \ref{antitone} simplifies to 
\begin{align*}
u \geq_K v \; \Rightarrow\; \nabla\psi(u) \leq_{K^*} \nabla\psi(v).
\end{align*}
That is, the gradient, as a map from $\mathbb{R}^n$ (equipped with $\leq_{K}$) to $\mathbb{R}^n$ (equipped with $\leq_{K^*}$),
is order-reversing (also known as \emph{antitone}). 
If $\psi$ satisfies Assumption \ref{antitone}, then so does $\psi \circ A$, when $A$ is a linear map such that $ \inner{y}{A v} \geq 0$,
for all $v \in K$ and $y \in K^*$. When $\psi$ is twice differentiable, Assumption \ref{antitone} is equivalent to
$ \inner{w}{\nabla^2 \psi(u) v} \leq 0$,
for all $u, v, w \in K$. For example, this is equivalent to Hessian being element-wise non-positive when $K = \mathbf{R}_{+}^n$. 
\subsection{Competitive ratio for non-decreasing $\psi$}
To quantify the competitive ratio of the algorithms, we define $\alpha_\psi: K \mapsto \mathbf{R}$ as
\begin{align}\label{alpha}
\alpha _\psi(u) = \sup{\{c\; | \;\psi^*(y) \geq c \psi(u),\; \forall y \in \partial \psi(u)\}},
\end{align}
which for $u\neq0$ can also be written as
\[
\alpha_\psi(u) = \inf_{y \in \partial \psi(u)}{\frac{\psi^*(y)}{\psi(u)}}.
\]
Since $\psi^*(y) + \psi(u) = \inner{y}{u}$ for all $y \in \partial \psi(u)$, $\alpha_\psi$ is equivalent to\begin{full}\footnote{If $\psi(0) \neq 0$, then the definition of $\alpha_\psi$ should be modified to $\alpha _\psi(u) = \sup\{c\; | \;\inner{y}{u}  \geq (c+1) (\psi(u)-\psi(0)), y \in \partial \psi(u)\}$.}\end{full}
\begin{align}\label{alpha2}
\alpha _\psi(u) &= \sup\{c\; | \;\inner{y}{u}  \geq (c+1) \psi(u),\; \forall y \in \partial \psi(u)\}\\ \notag
&= \sup\{c\; | \; \psi'(u;u)  \geq (c+1) \psi(u)\},
\end{align}
where $\psi'(u,v)$ is the directional derivative of $\psi$ at $u$ in the direction of $v$. Note that for any $u \in K$, we have $-1 \leq \alpha _\psi(u) \leq 0$ since for any $y \in \partial{\psi(u)}$, by concavity of $\psi$ and the fact that $y \in K^*$, we have
$0 \leq \inner{y}{u} \leq \psi(u) - \psi(0).$  If $\psi$ is a linear function then $\alpha_{\psi} = 0$, while if $0 \in \partial \psi(u)$ for some $u \in K$, then $\alpha_{\psi}(u) = -1$. 
\begin{full}
 Figure \ref{alpha-figure} shows a differentiable function $\psi: \mathbb{R}
_{+} \mapsto \mathbb{R}$. The value of $\alpha(u_0)$ is the ratio of $ \psi^*(\psi'(u_{0}))$ to $\psi(u_0)$ as shown.

\begin{figure}
\begin{center}
\begin{tikzpicture}[thick,scale=1, every node/.style={scale=1}]
  \draw[line width=1.5 pt,->] (-0.4,0) -- (3.7,0);
  \draw[line width=1.5 pt,->] (0,-.4) -- (0,2.3) node[above] {};
  \draw[line width=.75 pt,dashed] (1.5,0) -- (1.5,1.42);
  \draw[line width=.75 pt,dashed] (0,1.42) -- (1.5,1.42);
  \draw[line width=2 pt,scale=1,domain=0:3.4,smooth,variable=\x,blue] plot ({\x},{2*exp(-.2)+\x/10- 2*exp(-\x-.2) });
  \draw[scale=1,domain=0:3.4,smooth,variable=\x] plot ({\x},{(0.1+2*exp(-1.7))*\x+.75});
    \draw (0,.64) node[left]{$-\psi^*(\psi'(u_0))$};
    \draw (0,1.42) node[left]{$\psi(u_0)$};
    \draw (1.5,0) node[below]{$u_0$};
    \draw (4,0) node[below]{$u$};
\end{tikzpicture}
\end{center}
\caption{ $\psi'$ is the derivative of $\psi$ and $ -\psi^*(\psi'(u_{0}))$ is the $y$-intercept of the tangent to the function graph at $u_{0}$. ${\alpha}_{\psi}(u_0)$ is the ratio of $\psi^*(\psi'(u_{0}))$ to $\psi(u_0)$.}\label{alpha-figure}
\end{figure}
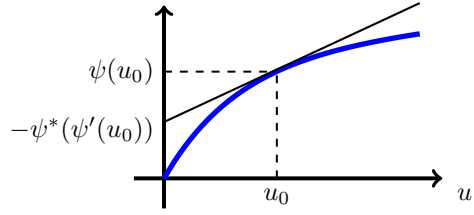
\end{full}
%

The next theorem provides lower bounds on the competitive ratio of the two algorithms. 

\begin{full}
\begin{thm}\label{thm1}
If Assumption \ref{antitone} holds, then for the simultaneous update algorithm we have
\footnote{If $\psi(0) \neq 0$, then the conclusion becomes $P_{\rm sim} - \psi(0)\geq \frac{1}{1-\alpha _\psi} (D^\star - \psi(0))$.}:
\begin{align*}
P_{\rm sim}\geq \frac{1}{1-\bar{\alpha}_\psi} D^\star.
\end{align*}
where $D^\ast$ is the dual optimal objective and
\begin{align}\label{alpha-bar}
 \bar{\alpha}_{\psi} = \inf\{\alpha_{\psi}(u)\,|\, u \in  K\}.
\end{align}
For the sequential update algorithm,
\begin{align*}
P_{\rm seq} \geq \frac{1}{1-\bar{\alpha} _\psi} (D^\star+ \sum_{t=1}^{m} \inner{ A_t \hat{x}_t}{ \hat{y}_{t+1} - \hat{y}_t}). 
\end{align*}
Further, if $\psi$ is differentiable with gradient Lipschitz continuity parameter $\frac{1}{\mu}$ with respect to $\norm{\cdot}$,
\begin{align}\label{Riemannian-sum}
P_{\rm seq}  \geq \frac{1}{1-\bar{\alpha} _\psi} (D^\star - \sum_{t=1}^{m} \frac{1}{2\mu}\norm{ A_t \hat{x}_t}^2).
\end{align}  
\end{thm} 
%
%
%
%

\begin{proof}
We first show that Assumption \ref{antitone} implies that
\begin{align}\label{tilde-sigma}
\sum_{t=1}^{m} \sigma_t(A_t^T \tilde{y}_t) \geq  \sum_{t=1}^{m} \sigma_t(A_t^T \tilde{y}_{m+1}),\\\notag
\sum_{t=1}^{m} \sigma_t(A_t \hat{y}_t) \geq \sum_{t=1}^{m} \sigma_{t}(A_t \hat{y}_{m+1}),
\end{align}

To do so, we prove that for all $t$, $\sigma_t(A_t^T \tilde{y}_t) \geq \sigma_{t}(A_t^T \tilde{y}_{m+1})$. The proof for $\sigma_t(A_t^T \hat{y}_t) \geq \sigma_{t}(A_t^T \hat{y}_{m+1})$ follows the same steps. 

For any $t$, we have $\sum_{s=1}^{t} A_s \tilde{x}_s \leq_{K} \sum_{s=1}^{m} A_s \tilde{x}_s$ since $A_s F_s \subset K$ for all $s$. Since $\tilde{y}_{t} \in \partial \psi(\sum_{s=1}^{t} A_s \tilde{x}_s)$ and $\tilde{y}_{m+1}$ was chosen to be the minimum element in $\partial\psi(\sum_{s=1}^{m} A_s \tilde{x}_s)$ with respect to $\leq_{K^*}$, by the Assumption \ref{antitone}, we have 
\begin{align*}
\tilde{y}_{t} \geq_{K^*} \tilde{y}_{m+1}.
\end{align*}

Since $A_{t} x \in K$ for all $x \in F_{t}$, we get $\inner{A_{t} x}{y_{t}} \geq \inner{A_{t} x}{y_{m}}. $
Therefore, $  \sigma_{t}(A_{t}^T \tilde{y}_{t}) \geq \sigma_{t}(A_{t}^T \tilde{y}_{m+1})  $.
Now by \eqref{tilde-sigma} 
\begin{align*}
D_{\rm sim} &=  \sum_{t=1}^{m} \sigma_t(A_t^T \tilde{y}_t) - \psi^{*}(\tilde{y}_{m+1}) \geq  \sum_{t=1}^{m} \sigma_t(A_t^T \tilde{y}_{m+1}) - \psi^{*}(\tilde{y}_{m+1}) \geq D^\star.
\end{align*}
Similarly, $D_{\rm seq} \geq D^\star$. Lemma \ref{duality-gap} and definition of $\bar{\alpha}_{\psi}$ give the desired result.
\end{proof}
\end{full}

\begin{full}
\subsection{Competitive ratio for non-monotone $\psi$}\label{non-monotone-section}
If $\psi$ is not non-decreasing with respect to $K$, the algorithms are not guaranteed to increase the objective at each step; therefore, $P_{\rm seq}$ and $P_{\rm sim}$ are not guaranteed to be non-negative. However, if we add the assumption that $0 \in F_t$ for all $t$,  then Algorithm \ref{algorithm2} will not decrease the objective at any step,
\begin{align}\label{non-decreasing-prop}
 \psi(\sum_{s=1}^{t} A_s \tilde{x}_s) - \psi(\sum_{s=1}^{t-1} A_s \tilde{x}_s) \geq \inner{\tilde{x}_t}{A_t^T \tilde{y}_t} = \max_{x \in F_t}\inner{x}{A_t^T \tilde{y}_t} \geq 0,
\end{align}
which follows from concavity of $\psi$.

In other words, the algorithm simply assigns $0$ to $x_t$ if any other feasible point in $F_t$ reduces the objective value. Define 
$$\tilde{u} = \sum_{t=1}^{m} A_t \tilde{x}_t,$$ 
and note that under the assumption $0 \in F_t$, we have $\psi(\tilde{u}) \geq 0$. Further, we have:
\begin{align}\label{non-monotone-alpha}
P_{\rm sim} \geq \frac{1}{1 - \alpha_{\psi}(\tilde{u})}D^\star.
\end{align}

We provide examples of non-monotone $\psi$ and their competitive ratio analysis in this section and section \ref{nesterov-section}. Given appropriate conditions on $F_t$ and $A_t$, in order to find a lower bound on the competitive ratio independent of $\tilde{u}$, we only need to lower bound $\alpha_{\psi}$ over a subset of $K$. Note that when $\psi$ is not non-decreasing, then there exists a supergradient that is not in $K^*$. Therefore, $\alpha_{\psi}(\tilde{u})$ for general $\psi$ can be less than $-1$. In this case, the lower bound for the competitive ratio of Algorithm \eqref{algorithm2} is less than $\frac{1}{2}$.

\end{full}

%
%

We now consider examples of $\psi$ that satisfy Assumption \ref{antitone} and derive lower bound on ${\alpha}_\psi$ for those examples.

\subsection{Examples on the positive orthant.}\label{orthant-examples}
Let $K = \mathbb{R}_{+}^{n}$ and note that $K^* = K$. To simplify the notation we use $\leq$ instead of $\leq_{\mathbb{R}^n_+}$. When $\psi$ has continuous partial second derivatives with respect to all the variables, Assumption \ref{antitone} is equivalent to the Hessian being element-wise non-positive over $\mathbb{R}^n_+$. Assumption \ref{antitone} is closely related to submodularity. In fact, on the positive orthant this assumption is sufficient for submodularity of $\psi$ on the lattice defined by $\leq$. However, this assumption is not necessary for submodularity. When $\psi$ has continuous partial second derivatives with respect to all the variables, the necessary and sufficient condition only requires the off-diagonal elements of the Hessian to be non-positive \cite{lorentz1953inequality}.

If $\psi$ is separable, i.e.,\begin{nips} $\psi(u) = \sum_{i=1}^{n} \psi_{i}(u_i),$\end{nips}
\begin{full}\begin{align}\label{sum}
\psi(u) = \sum_{i=1}^{n} \psi_{i}(u_i),
\end{align}
\end{full}
Assumption \ref{antitone} is satisfied since by concavity for each $\psi_i$ we have $\partial \psi_i(u_i) \leq \partial \psi_i(v_i)$ when $u_i \leq v_i$.  If $\psi$ satisfies the Assumption \ref{antitone}, then so does $\psi(A u )$, where $A$ is an element-wise non-negative matrix. 
\paragraph{Adwords} In the basic adwords problem, for all $t$, $F_t = \{x \in \mathbb{R}^l_+\;|\; \mathbf{1}^T x \leq 1 \}$, $A_t$ is a diagonal matrix with non-negative entries, and
\vspace{-1ex}
\begin{align}\label{adwords-objective}
\psi(u) = \sum_{i=1}^{n} u_i - \sum_{i=1}^{n}(u_i -1)_{+},
\end{align}
where $(\cdot)_+ = \max\{\cdot,0\}$. 
In this problem, $\psi^*(y) =  \mathbf{1}^T (y - \mathbf{1})$. Since $0 \in \partial \psi(\mathbf{1})$ we have $\alpha_{\psi} = -1$ by \eqref{alpha2}; therefore, the competitive ratio of algorithm \ref{algorithm2} is $\frac{1}{2}$. Let $r = \max_{t,i,j}{{A}_{t,i j}}$, then we have $\sum_{t=1}^{m} \inner{ A_t \hat{x}_t}{ \hat{y}_{t+1} - \hat{y}_t}  \leq n r $. Therefore, the competitive ratio of algorithm \ref{algorithm1} goes to $\frac{1}{2}$ as $r$ (bid to budget ratio) goes to zero. In adwords with concave returns studied in \cite{devanur2012online}, $A_t$ is diagonal for all $t$ and $\psi$ is separable \footnote{Note that in this case one can remove the assumption that $\partial \psi_i \subset \mathbb{R}_+$ since if $\tilde{y}_{t,i} = 0$ for some $t$ and $i$, then $\tilde{x}_{s,i} = 0$ for all $s \geq t$. }.

\paragraph{Online linear program and non-separable budget} Recall that online LP is given by
\begin{align*}
 \mbox{maximize}&{  \quad \sum_{t=1}^{m} c_t^T x_t}+ I_{\{\cdot \leq \mathbf{1}\}}\left(\sum_{t=1}^{m} B_t x_t\right)\\ \notag
& \quad  x_t \in F_t,  \quad  \forall t \in [m].
\end{align*}
where $F_t = \{x \in \mathbf{R}_+^k ~|~ \mathbf{1}^T x \leq 1\}$ is the simplex. If a lower bound on the optimal dual variable $\min_i y^*_{i} \geq - l$ is given, then the LP can be written in the exact penalty form \cite{bertsekas1975necessary}:
\begin{align}\label{General-lp}
\mbox{maximize}_{x_t \in F_t} \sum_{t=1}^{m} c_t^T x_t + G\left(\sum_{t=1}^{m} B_t x_t\right).
\end{align}
where $G$ is an $l$-Lipschitz continuous function with respect to $l_1$ norm given by
\begin{align*}
 G(u) = -l \sum_{i=1}^{n}(u_i-1)_{+}.
\end{align*}
$-l$ is a lower bound on dual variable if 
\begin{align}\label{l-bound}
l > \max \left\{\frac{c_{t,j}}{B_{t,i j}}\;|\; B_{t,i j} > 0, \; j \in [k],\; i \in [n]\right\}.
%
\end{align} 
  To prove this, by way of contradiction, we assume that ${y}_{i}^* \leq - l$ for some $i \in [n]$. Now by the definition of $l$, we have $c_{t,j} + (B_{t}^T y^*)_j < 0$ for all $j$ such that $B_{t, i j} > 0$. On the other hand by the optimality condition \eqref{primal-optimality}, $c_t + B_{t}^T y^* \in N_{F_t}(x_t^*)$ where $N_{F_t}(x_t^*)$ is the normal cone of the simplex at $x_t^*$. Therefore, we should have $x^*_{t,j} = 0$ for all $j$ such that $B_{t, i j} > 0$. This results in $(B_{t} {x}_{t}^{*})_{i} = 0$ which yield $(\sum_{t=1}^{m} B_{t} x_t^*)_i = 0$. This means that the corresponding variable $y_i^* = 0$ which is a contradiction. With this choice of $l$, Algorithm \ref{algorithm2} always maintains a feasible solution for the original problem when applied to \eqref{General-lp}. This follows from the fact that if $\tilde{y}_{t,i} = -l$ for some $t$ and $i$, then $(B_t \tilde{x}_t)_i = 0$. Note that in the adwords problem $ B_t = \diag(c_t)$ and $l = 1$. 
  

For any $p \geq 1$ let $\mathcal{B}_p$ be the $l_p$-norm ball. In order to provide examples of non-separable $G$, we rewrite the function $\sum_{i=1}^{n}(u_i-1)_{+}$ using the distance from $\mathcal{B}_\infty$ . For any set $C \subset \mathbb{R}_{+}^n$, let $d_1(u,C) = \inf_{\bar{u} \in C} \norml{u - \bar{u}}$ and note that $d_1(u,C)$ is $1$-Lipschitz continuous with respect to $l_1$ norm. We have:
\begin{align}\label{infty-distance}
\sum_{i=1}^{n}(1-u_i)_{+} = d_1(u,\mathcal{B}_{\infty}).
\end{align}

 Consider a problem with constraint $\norm{\sum_{t=1}^{m} B_t x_t }_{p} \leq 1$. Given the bound on dual variable in \eqref{l-bound}, this problem can be equivalently written in the form of \eqref{General-lp} with the exact penalty \cite[Theorem 5.5]{burke1991exact}
\begin{align*}
G(u) = - l d_1(u,\mathcal{B}_{p}).
\end{align*}
When $p = \infty$, we get back \eqref{infty-distance}.  

For $p \geq 1$, although not separable, the function $G(u) = -d_1(u,\mathcal{B}_{p})$ satisfies Assumption \ref{antitone}.  For $\mathcal{B}_1$, that follows from the fact that $d_1(u, \mathcal{B}_1) = (\mathbf{1}^T u- 1)_+$; therefore, $\partial d_1(u,\mathcal{B}_1) = \mathbf{1} \partial f(\mathbf{1}^T u)$, where $f(x) = (x- 1)_{+} $.  The proof for $p>1$ is given in Appendix \ref{lin-lp}. 

When $\psi$ is given by \ref{General-lp} with $G(u) = - l d_1(u,\mathcal{B}_{p})$, we have ${\alpha}_{\psi}(\tilde{u}) \geq -\frac{l}{\theta}$, where  $\theta =\min_{t} \min_{x \in F_t} \frac{c_t^T x}{\mathbf{1}^T B_{t} x}$. The derivation of this bound is also given in Appendix \ref{lin-lp}.

\subsection{Examples on the positive semidefinite cone.}\label{sec:logdet}
 Let $K = S_{+}^n$ and note that $K^* = K$. 
  An interesting example that satisfies Assumption \ref{antitone} is $\psi(U) = \tr{U^p}$ with $p \in (0,1)$, where $\nabla \psi(U) =p  U^{p-1}$ and $\bar{\alpha}_{\psi} = p - 1$. This objective function is used in $p \text{th}$ mean optimal experiment design.
 Another example is $\psi(U) = \log\det(U+A_0)$, where  $\nabla \psi(U) = (U+A_0)^{-1}$ and $\bar{\alpha}_\psi = -1$ since
%
 \begin{align*}
 \frac{\inner{\nabla \psi(U)}{U}}{\psi(U) - \psi(0)} = \frac{n - \tr{(A_0 + U)^{-1} A_0}}{\log\det((A_0 + U) A_0^{-1})} \rightarrow 0 \mbox{   as    }  \tr{U} \rightarrow \infty .
 \end{align*}
 
 
Maximizing $\log\det$ arises in several offline applications including D-optimal experiment design \cite{pukelsheim1993optimal}, maximizing the Kirchhoff complexity of a graph \cite{zavlanos2008distributed}, Optimal sensor selection \cite{Joshi2009,Shamaiah2010}. \begin{full}An example of applications for maximizing the logdet of a projected Laplacian, as a function of graph edge weights, appears in connectivity control of mobile networks \cite{zavlanos2008distributed}. \end{full}
\begin{full}
We derive the competitive ratio of the greedy algorithm with smoothing for online experiment design and Kirchhoff complexity maximization of a graph in section \ref{s::smoothing}.
\end{full}
\begin{nips}
We have analyzed the competitive ratio of the greedy algorithm with smoothing for online experiment design and Kirchhoff complexity maximization of a graph, which due to space limitation is given in the supplementary materials.
\end{nips}

\section{Smoothing of $\psi$ for improved competitive ratio}\label{s::smoothing}

The technique of ``smoothing'' an objective function, or equivalently adding a strongly convex regularization term to its conjugate function, has been used in several areas. In convex optimization, a general version of this is due to Nesterov \cite{nesterov2005smooth}, and has led to faster convergence rates of first order methods for non-smooth problems. 
%
In this section, we study how replacing $\psi$ with a appropriately smoothed function ${\psi_S}$ helps improve the performance of the two algorithms discussed in section \ref{two-algorithms}, 
and show that it provides optimal competitive ratio for two of the problems mentioned in section \ref{CA}, adwords and online LP. 
We then show how to maximize the competitive ratio of Algorithm \ref{algorithm2} for a separable $\psi$ and compute the optimal smoothing by solving a convex optimization problem. This allows us to \emph{design} the most effective smoothing
customized for a given $\psi$: we maximize the bound on the competitive ratio over the set of smooth functions.
(see subsection \ref{optimal-smoothing} for details). 


Let ${\psi_S}$ denote an upper semi-continuous concave function (a smoothed version of $\psi$), and suppose ${\psi_S}$ satisfies Assumption \ref{antitone}. The algorithms we consider in this section are the same as Algorithms \ref{algorithm1} and \ref{algorithm2}, but with $\psi$ replacing ${\psi_S}$. Note that the competitive ratio is computed with respect to the original problem, that is the offline primal and dual optimal values are still the same $P^\star$ and $D^\star$ as before. If we replace $\psi$ with ${\psi_S}$ in algorithms \ref{algorithm1} and \ref{algorithm2}, the dual updates are modified to
\begin{align*} 
&\hat{y}_{t+1} \in \argmin_{y} \inner{\sum_{s=1}^{t} A_s \hat{x}_s}{y} - {\psi_S}^*(y) ,\\
&(\tilde{x}_t, \tilde{y}_t) \in \arg\min_{y}\max_{x \in F_t} \;\inner{y}{A_t x + \sum_{s=1}^{t-1} A_s x_s} - {\psi_S}^*(y).
\end{align*}

From Lemma \ref{duality-gap}, we have that 
\begin{align}
&D_{\rm sim} \leq {\psi_S}\left(\sum_{t=1}^{m} A_t \tilde{x}_t\right) - \psi^*(\tilde{y}_{m+1}) \\ \label{smooth_duality_gap}
&D_{\rm seq} \leq {\psi_S}\left(\sum_{t=1}^{m} A_t \hat{x}_t\right)-  \psi^*(\hat{y}_{m+1}) - \sum_{t=1}^{m} \inner{ A_t \hat{x}_t}{ \hat{y}_{t+1} - \hat{y}_t}.
\end{align}
%
Similar to our assumption on $\psi(0)$, to simplify the notation, by replacing ${\psi_S}(\cdot)$ with ${\psi_S}(\cdot) - {\psi_S}(0)$, we assume ${\psi_S}(0) = 0$. Define
\begin{align*}
&\alpha_{\psi,{\psi_S}}(u) =\sup \{c \;|  \psi^*(y) \geq  {\psi_S}(u) + (c-1) \psi(u)  ,  \forall y \in \partial {\psi_S}(u)\},
\end{align*}
and
\begin{align*}
\bar{\alpha} _{\psi,{\psi_S}} = \inf\{\alpha _{\psi,{\psi_S}}(u) \; | \; u \in K\}.
\end{align*}

Now the conclusion of Theorem \ref{thm1} holds with $\bar{\alpha} _{\psi}$ replaced by $\bar{\alpha} _{\psi,{\psi_S}}$. Similarly, when ${\psi_S}$ is non-monotone, inequality \eqref{non-monotone-alpha} holds with ${\alpha} _{\psi}$ replaced by ${\alpha} _{\psi,{\psi_S}}$.
%
%
%

\subsection{Nesterov Smoothing}
\label{nesterov-section}

We first consider Nesterov smoothing \cite{nesterov2005smooth}, and apply it to examples on non-negative orthant. 
Given a proper upper semi-continuous concave function $\phi: \mathbb{R}^n \mapsto \mathbb{R} \cup \{-\infty\}$, let
\begin{align*}
{\psi_S} = (\psi^*+\phi^*)^*.
\end{align*}
Note that ${\psi_S}$ is the hypo-sum (sup-convolution) of $\psi$ and $\phi$.
\begin{align*}
{\psi_S} = \psi \Box \phi (u) = \sup_{v} \psi(v) + \phi(u-v).
\end{align*}
This is called hypo-sum of $\psi$ and $\phi$ since the hypo-graph of ${\psi_S}$ is the Minkowski sum of hypo-graphs of $\psi$ and $\phi$. 

If $\psi$ and $\phi$ are separable, then $\psi \Box \phi$ satisfies Assumption \ref{antitone} for $K = \mathbb{R}_{+}^n$. Here we provide example of Nesterov smoothing for functions on non-negative orthant.

\paragraph{Adwords and Online LP:}
Consider the problem \eqref{General-lp} with $G(u) = -l \sum_{i=1}^{n}(u_i-1)_{+}$. For this problem we smooth $G$ with:
\begin{align*}
\phi^*(y)=&\frac{1}{\gamma}(\sum_{i=1}^{m} ((y_i -\frac{\theta}{(e-1)}) \log(1-\frac{(e-1)}{\theta} y_i))-(1+\gamma)\mathbf{1}^T y),
\end{align*}
\noindent where $\gamma = \log(1+ \frac{l (e-1)}{\theta})$, 
\begin{align*}\theta =\min_{t} \min_{x \in F_t} \frac{c_t^T x}{\mathbf{1}^T B_{t} x}\end{align*} and $l$ is defined as in \eqref{l-bound}. In this case, we have;
\begin{align*}
\alpha_{\psi,\psi\Box\phi}(\tilde{u}) \geq 1 - (1+\frac{1}{e-1})\gamma .
\end{align*}
(see the Appendix \ref{alpha-derivation} for the derivation). This gives the competitive ratio of $\frac{1}{\gamma}(1-1/e) $. In the case of adwords where $\theta = l$, this yield the optimal competitive ratio of $1-e^{-1}$ and the smoothed function coincides with the one derived in the previous paragraph. For a general LP, this approach provides the optimal competitive ratio, which is known to be $O(\gamma^{-1})$ \cite{buchbinder2009online}. 

\paragraph{Online experiment design and online graph formation:}  
As pointed in section \ref{sec:logdet}, maximizing the determinant arises in graph formation, sensor selection and optimal experiment design. In the offline form, these problems can be cast into maximizing $\log\det$ subject to linear constraints as follows:
\begin{align}\label{Kirchhoff}
\mbox{maximize}  \quad &\log\det(A_0+\sum_{t=1}^{m} a_t a_t^T x_t)\\ \notag
\mbox{subject to} \quad &\sum_{t=1}^{m}  x_t \leq b,\\ \notag
& x_t \in [0,1], \quad \forall t,
\end{align} 
where $A_0 \succ 0$ and  $a_t \in \mathbb{R}^n$. We describe the online version of this problem for graph formation and experiment design case. Consider the following online problem on a graph: given an underlying connected graph with the Laplacian matrix $L_0$, at round $t$, the online algorithm is presented with an edge and should decide whether to pick or drop the edge by choosing $x_t \in \{0,1\}$. There is a bound on the number of the edges that can be picked, $\sum_{t=1}^{n} x_t \leq b$. The goal is to maximize the Kirchhoff complexity of the graph. The convex relaxation of the optimization problem is given by \ref{Kirchhoff}
where $A_0 = L_0 + \mathbf{1}\mathbf{1}^T$ and $a_t$ is the incidence vector of the edge presented at round $t$. 
The online experiment design or sensor selection is the same as \eqref{Kirchhoff}. In that problem the vectors $a_t$ are experiment or measurement vectors.

The dual to \ref{Kirchhoff} is as follows
\begin{align*}
\mbox{maximize} \quad &\sum_{t=1}^{m} (a_t^T Y a_t + y)_+ + n-\log\det(Y) +\tr{A_0 Y}-b y\\
\mbox{subject to} \quad& y\leq 0, \;\; Y \geq 0.
\end{align*} 


The hard budget constraint in problem \ref{Kirchhoff} can be replaced by the exact penalty function $G(u) = - l (u-b)_{+}$ when $-l$ is a given lower bound on the optimal dual variable corresponding to the constraint. $-l$ is a lower bound on the optimal dual variable corresponding to the constraint if
\begin{align}\label{Kirchhoff-l}
l > \frac{2}{\lambda_{\min}(A_0)}.
\end{align}

This follows from the fact that by the optimality condition \eqref{dual-optimality}, $Y^* = (A_0+\sum_{t=1}^{n} a_t a_t^T {x}_t^*)^{-1} $. If $y^* \leq -l$, then $a_t^T Y^* a_t + y^* < 0$ for all $t$. This combined with the the primal optimality condition \eqref{primal-optimality} forces $x_t^* = 0$ for all $t$. Therefore, we get $\sum_{t=1}^{m} x_t^* = 0$ and $y^* = 0$, which is a contradiction. 
Similar to the previous problems, we smooth $G$ with:
\begin{align*}
\phi^*(y)=\frac{b}{\gamma} ((y -\frac{\theta}{e-1}) \log(1-\frac{e-1}{\theta} y)-(1+\gamma)y).
\end{align*}
Here $\gamma = \log(1+\frac{l}{\theta})$ and $\theta = \log(1+1/n)$. For this problem we have:
\begin{align*}
\alpha_{\psi,\psi\Box\phi}(\tilde{u}) \geq 1 - (1+\frac{1}{e-1})\gamma + \bar{\alpha}_{\log\det}\geq  - (1+\frac{1}{e-1})\gamma.
\end{align*}
This results in the following competitive ratio bound
\begin{align*}
P_{\rm sim} - \log\det(A_0) &\geq  \frac{1}{1+\left(1+\frac{1}{e-1}\right)\gamma} (D^\star - \log\det(A_0)).
\end{align*}

The proof is included in the Appendix \ref{alpha-derivation}. 
In the graph formation problem, 
$\lambda_{\min} (A_0) = \lambda_{2}(L_0)$, that is the second smallest eigenvalue of $L_0$. For a connected graph $\lambda_2(L_0) \in [1/n,\; n]$ with lower bound achieved by the path graph and the upper bound achieved by the complete graph. Therefore, the competitive ratio in this case is $\Omega(\frac{1}{\log n})$.

\subsection{Computing optimal smoothing for separable functions on $\mathbb{R}_{+}^n$}
\label{optimal-smoothing}

We now tackle the problem of finding the optimal smoothing for separable functions on the positive orthant, which as we show in an example at the end of this section is not necessarily given by Nesterov smoothing. 
Given a separable monotone $\psi(u) = \sum_{i=1}^{n}\psi_i(u_i)$ and ${\psi_S}(u) = \sum_{i=1}^{n}{\psi_S}_i(u_i)$ on $\mathbb{R}_{+}^n$ we have that
$\bar{\alpha}_{\psi,{\psi_S}} \geq \min_{i}\bar{\alpha}_{\psi_i,{\psi_S}_i}$.

To simplify the notation, drop the index $i$ and consider $\psi: \mathbb{R}_{+} \mapsto \mathbb{R}$.
We formulate the problem of finding ${\psi_S}$ to maximize $\alpha_{\psi,{\psi_S}}$ as an optimization problem. In section \ref{discussion} we discuss the relation between this optimization method and the optimal algorithm presented in \cite{devanur2012online}. 
We set ${{\psi}_{M}}(u) = \int_{0}^{u} y(s) ds$ with $y$ a continuous function, and state the infinite dimensional convex optimization problem with $y$ as a variable,
\begin{align}\label{cont-problem}
 \mbox{minimize}&{ \quad  \beta}\\ \notag
\mbox{subject to}& \quad  \int_{0}^{u} y(s) ds  - \psi^*(y(u)) \leq \beta  \psi(u), \qquad \forall u \in [0,\infty)\\ \notag
&\quad y \in C[0,\infty).
\end{align}

%
where $\beta=1-\bar{\alpha}_{\psi,{\psi_S}}$ (theorem \ref{thm1} describes the dependence of the competitive ratios on this parameter).  Note that we have not imposed any condition on $y$ to be non-increasing (i.e., the corresponding ${\psi_S}$ to be concave). The next lemma establishes that every feasible solution to the problem \eqref{cont-problem} can be turned into a non-increasing solution. 
 \begin{lemma}\label{non-increasing-solution}
 Let $(y,\beta)$ be a feasible solution for problem \eqref{cont-problem} and define $\bar{y}(u) = \inf_{s \leq u} y(s)$. Then $(\bar{y}, \beta)$ is also a feasible solution for problem \eqref{cont-problem}. 
  \end{lemma} 
In particular if $(y,\beta)$ is an optimal solution, then so is $(\bar{y}, \beta)$. The proof is given in the supplement.
Revisiting the adwords problem, we observe that the optimal solution is given by $y(u) = \left(\frac{e - \exp(u)}{e-1}\right)_{+}$, which is the derivative of the smooth function we derived using Nesterov smoothing in section \ref{nesterov-section}. The optimality of this $y$ can be established by providing a dual certificate, a measure $\nu$ corresponding to the inequality constraint, that together with $y$ satisfies the optimality condition. If we set $d \nu = f(u) \;du$ with $f(u) = \exp{(1-u)}/(e-1)$, the optimality conditions are satisfied with $\beta = (1-1/e)^{-1}$.\begin{full}
\begin{align*}
& \int_{u=0}^{\infty} f(u) \psi(u) \; du= 1 ,\;\;  f \succeq 0,\\ 
 &\int_{s = u}^{\infty} f(s)  \; ds\in f(u) \partial \psi^*(y(u)),  \qquad \forall u \geq 0,\\ 
& \int_{0}^{u} y(s) ds  - \psi^*(y(u)) \leq \beta \psi(u), \qquad  \forall u \geq 0,\\ 
 & f(u) (\int_{0}^{u} y(s) ds  - \psi^*(y(u)) - \beta  \psi(u)) = 0,\qquad  \forall u \geq 0.
\end{align*} 
\end{full}

Also note that if $\psi$ plateaus (e.g.,  as in the adwords objective), then one can replace problem \eqref{cont-problem} with a problem over a finite horizon.

 \begin{thm}\label{plateau}
 Suppose $\psi(u) = c$ on $[u' , \infty)$. 
 Then problem \eqref{cont-problem} is equivalent to the following problem,
\begin{align}\label{finite-cont-problem}
 \mbox{\rm minimize}&{ \quad  \beta}\\ \notag
\mbox{\rm subject to}& \quad  \int_{0}^{u} y(s) ds  - \psi^*(y(u)) \leq \beta  \psi(u), \qquad \forall u \in [0,u']\\ \notag
& y(u') = 0, \quad y \in C[0, u'].
\end{align}

\end{thm}
%
So for a function $\psi$ with a plateau, one can discretize problem \eqref{finite-cont-problem} to get a finite dimensional problem,

\begin{align}\label{disc-problem}
 \mbox{minimize}& \quad \beta \\ \notag
\mbox{subject to}& \quad h \sum_{s=1}^{t} y[s]  - \psi^*(y[t]) \leq \beta  \psi(h t), \qquad \forall t \in [d]\\ \notag
& \quad y[d] = 0, 
\end{align}

where $h = {u'}/{d}$ is the discretization step. Figure \ref{fig-pl} shows the optimal smoothing for the piecewise linear function $\psi(u) = \min(.75,\; u,\; .5 u+.25)$ by solving problem \eqref{disc-problem}. We point out that the optimal smoothing for this function is \emph{not} given by Nesterov's smoothing (even though the optimal smoothing can be derived by Nesterov's smoothing for a piecewise linear function with only two pieces, like the adwords cost function). 
Figure \ref{fig-psistardiff} shows the difference between the conjugate of the optimal smoothing function and $\psi^*$ for the piecewise linear function, which we can see is not concave. 

In cases where a bound $u_{\max}$ on $\sum_{t=1}^{m} A_t F_t$ is known, we can restrict $t$ to $[0, u_{\max}]$ and discretize problem \eqref{cont-problem} over this interval. However, the conclusion of Lemma \ref{non-increasing-solution} does not hold for a finite horizon and we need to impose additional linear constraints  $y[t] \leq y[t-1]$ to ensure the monotonicity of $y$. We find the optimal smoothing for two examples of this kind: $\psi(u) = \log(1+u)$ over $[0, 100]$ (Figure \ref{fig-log}), and $\psi(u) = \sqrt{u}$ over $[0, 100]$ (Figure \ref{fig-root}). In Figure \ref{fig-logComp}, we show the competitive ratio achieved with the optimal smoothing of $\psi(u) = \log(1+u)$ over $[0, u_{\max}]$ as a function of $u_{\max}$. Figure \ref{fig-rootComp} depicts this quantity for $\psi(u) = \sqrt{u}$.
 \begin{figure}
 \begin{subfigure}{0.33\textwidth}
\centering    
                \includegraphics[width=.80\textwidth]{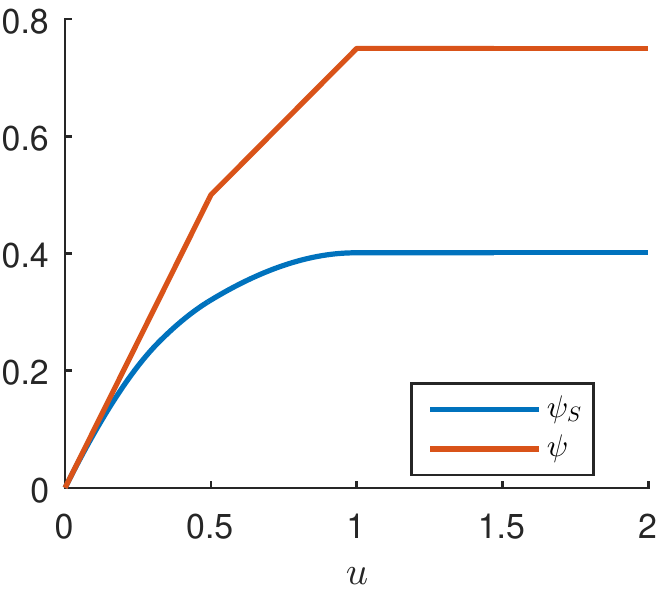}        
           \caption{}
              \label{fig-pl}        
 \end{subfigure}
  \begin{subfigure}{0.33\textwidth}
\centering    
                \includegraphics[width=.80\textwidth]{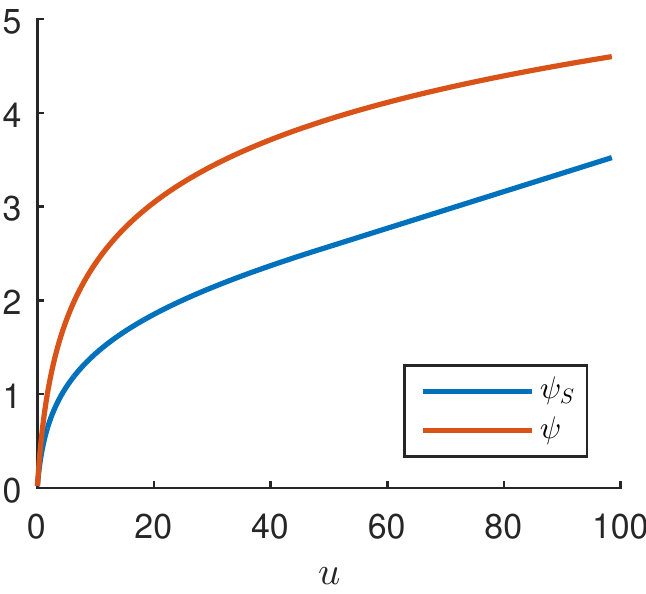}        
            \caption{}
              \label{fig-log}     
 \end{subfigure}
 \begin{subfigure}{0.33\textwidth}
\centering    
                \includegraphics[width=.80\textwidth]{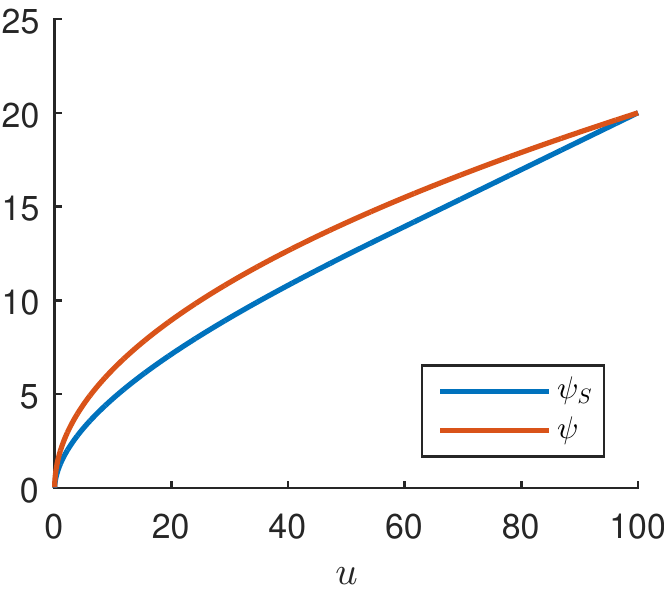}   
                 \caption{}
              \label{fig-root}               
 \end{subfigure}
 \begin{subfigure}{0.33\textwidth}
\centering    
                \includegraphics[width=.80\textwidth]{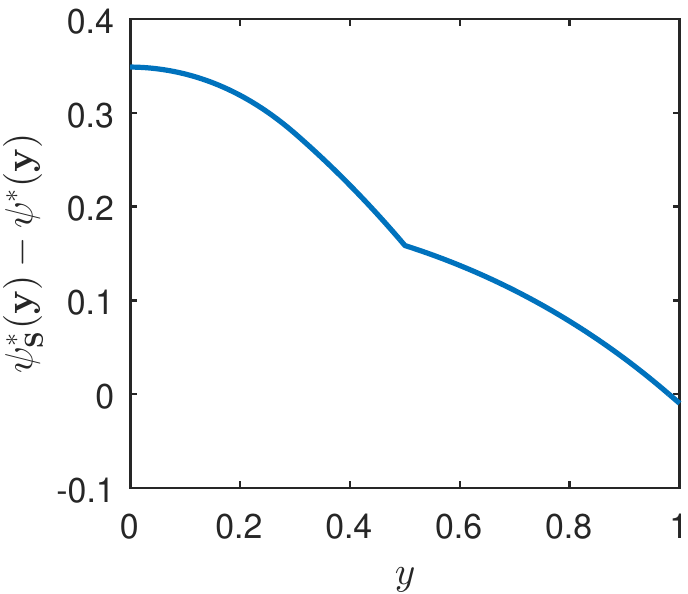}        
           \caption{}
              \label{fig-psistardiff}          
 \end{subfigure}
  \begin{subfigure}{0.33\textwidth}
\centering    
                \includegraphics[width=.80\textwidth]{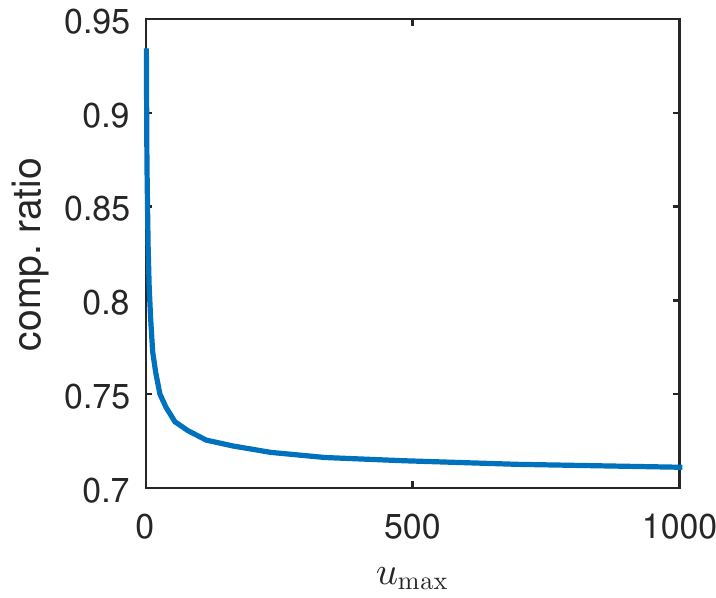}        
           \caption{}
              \label{fig-logComp}           
 \end{subfigure}
 \begin{subfigure}{0.33\textwidth}
\centering    
                \includegraphics[width=.80\textwidth]{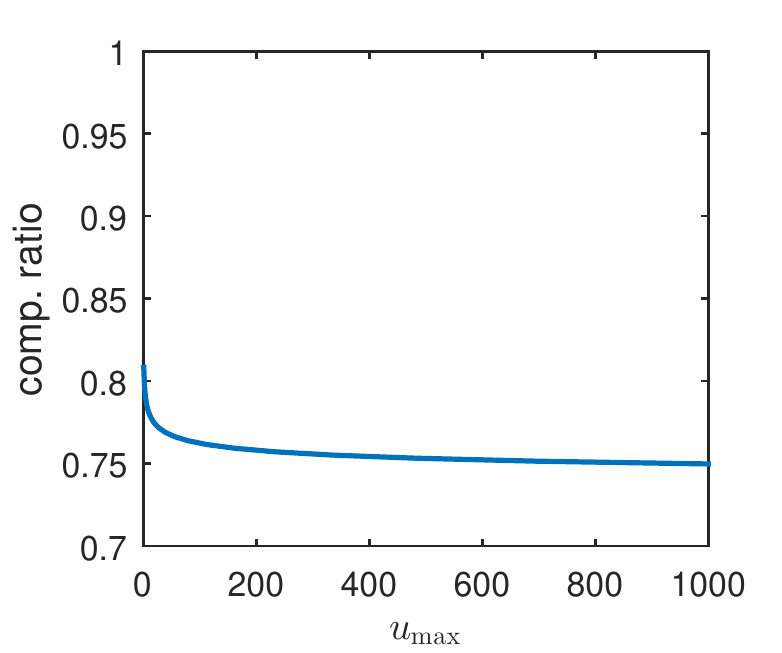}   
           \caption{}
              \label{fig-rootComp}         
 \end{subfigure}
 \caption{Optimal smoothing for $\psi(u) = \min(.75,\; u,\; .5 u+.25)$  (a), $\psi(u) = \log(1+u)$ over $[0, 100]$ (b), and $\psi(u) = \sqrt{u}$ over $[0, 100]$ (c). The competitive ratio achieved by the optimal smoothing as a function of $u_{\max}$ for $\psi(u) = \log(1+u)$ (e) and $\psi(u) = \sqrt{u}$ (f). ${\psi_S}^* - \psi^*$ for the piecewise linear function (d). }
\end{figure}

\subsection{Bounds for the sequential algorithm}
In this section we provide a lower-bound on the competitive ratio of the sequential algorithm (Algorithm \ref{algorithm1}). Based on this competitive ratio bound we modify Problem \eqref{cont-problem} for designing the smoothing function for the sequential algorithm. 

\begin{thm} Suppose ${\psi_S}$ is differentiable on an open set containing $K$ and satisfies Assumption~\ref{antitone}.  In addition suppose there exists $c \in K$ is such that $A_t F_t \leq_{K} c$ for all $t$, then
\begin{align*}
P_{\rm seq} \geq \frac{1}{1-\bar{\alpha}_{\psi,{\psi_S}} + \kappa_{c,\psi,{\psi_S}}} D^\star,
\end{align*}
where $\kappa$ is given by
\begin{align}\label{kappa}
\kappa_{c,\psi,{\psi_S}} = \inf \{r ~|~ \inner{c}{\nabla {\psi_S}(0) - \nabla {\psi_S}(u)} \leq r \psi(u) , u \in K\}
\end{align}

\end{thm}
\begin{proof}
Since ${\psi_S}$ satisfies Assumption \ref{antitone}, we have $\hat{y}_{t+1} \leq_{K^*} \hat{y}_t$. Therefore, we can write:
\begin{align}\notag 
\sum_{t=1}^{m} \inner{A_t \hat{x}_t}{\hat{y}_t -\hat{y}_{t+1}  } &\leq \sum_{t=1}^{m} \inner{c}{\hat{y}_t -\hat{y}_{t+1}  } \\ \label{local_bound}
&=  \inner{c}{\hat{y}_0 -\hat{y}_{m+1}  }
\end{align}
Now by combining \ref{smooth_duality_gap} with \ref{local_bound}, we get
\begin{align*}
&D_{\rm seq} \leq {\psi_S}\left(\sum_{t=1}^{m} A_t \hat{x}_t\right) + \inner{c}{\nabla {\psi_S}(0) - \nabla {\psi_S}\left(\sum_{t=1}^{m} A_t \hat{x}_t\right)}.
\end{align*}
The conclusion of the theorem follows from the definition of $\bar{\alpha}_{\psi,{\psi_S}}$, $\kappa_{c,\psi,{\psi_S}}$ and the fact that $\hat{D} \geq D^\star$.

\end{proof}

Based on the result of the previous theorem we can modify the optimization problem set up in Section \ref{optimal-smoothing} for separable functions on $\mathbf{R}_+^n$ to maximize the lower bound on the competitive ratio of the sequential algorithm. Note that in this case we have $\kappa_{c,\psi,{\psi_S}} \leq \max_{i}\kappa_{c_i,\psi_i,{\psi_S}_i} $. Similar to the previous section to simplify the notation we drop the index $i$ and assume $\psi$ is a function of a scalar variable. The optimization problem for finding $\psi_S$ that minimizes $\kappa_{c,\psi,{\psi_S}} -\bar{\alpha}_{\psi,{\psi_S}} $ is as follows:
\begin{align}\label{cont-problem-seq}
 \mbox{minimize}&{ \quad  \beta}\\ \notag
\mbox{subject to}& \quad  \int_{0}^{u} y(s) ds  + c (\psi'(0) - y(u))- \psi^*(y(u)) \leq \beta  \psi(u), \qquad \forall u \in [0,\infty)\\ \notag
&\quad y \in C[0,\infty).
\end{align}

In the case of Adwords, the optimal solution is given by 
$$\beta = \frac{1}{1 - \exp(\frac{-1}{c+1})}, \quad y(u) = \beta \left(1-\exp{\left(\frac{u-1}{1+c}\right)} \right)_{+},$$
which gives a competitive ratio of $1 - \exp\left(\frac{-1}{c+1}\right).$ In Figure \ref{fig-seq-log} we have plotted the competitive ratio achieved by solving problem \ref{cont-problem-seq} for $\psi(u) = \log\det(1+u)$ with $u_{\max} = 100$ as a function of $c$. Figure \ref{fig-seq-pl} shows the competitive ratio as a function of $c$ for the piecewise linear function $\psi(u) = \min(.75,\; u,\; .5 u+.25)$.

 \begin{figure}
 \begin{subfigure}{0.47\textwidth}
\centering    
                \includegraphics[width=.80\textwidth]{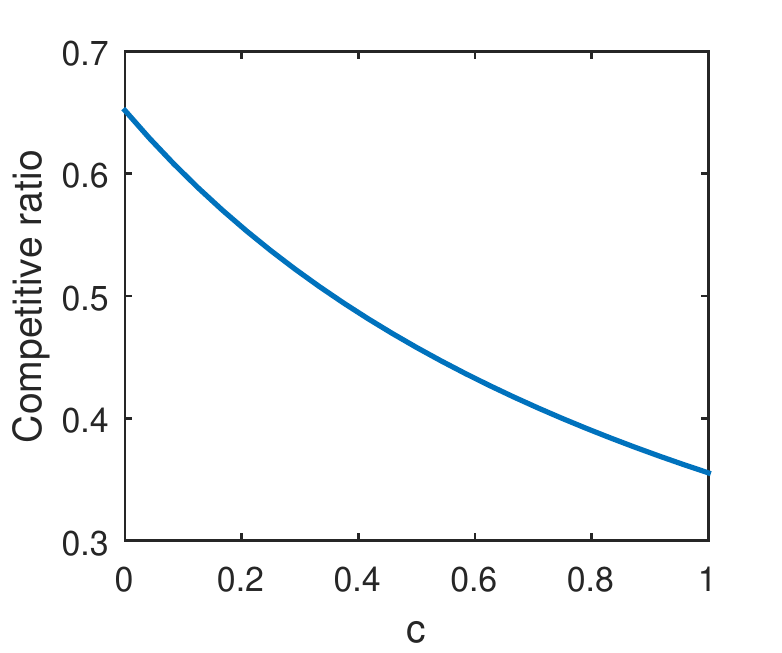}        
           \caption{}
              \label{fig-seq-pl}        
 \end{subfigure}
  \begin{subfigure}{0.47\textwidth}
\centering    
                \includegraphics[width=.80\textwidth]{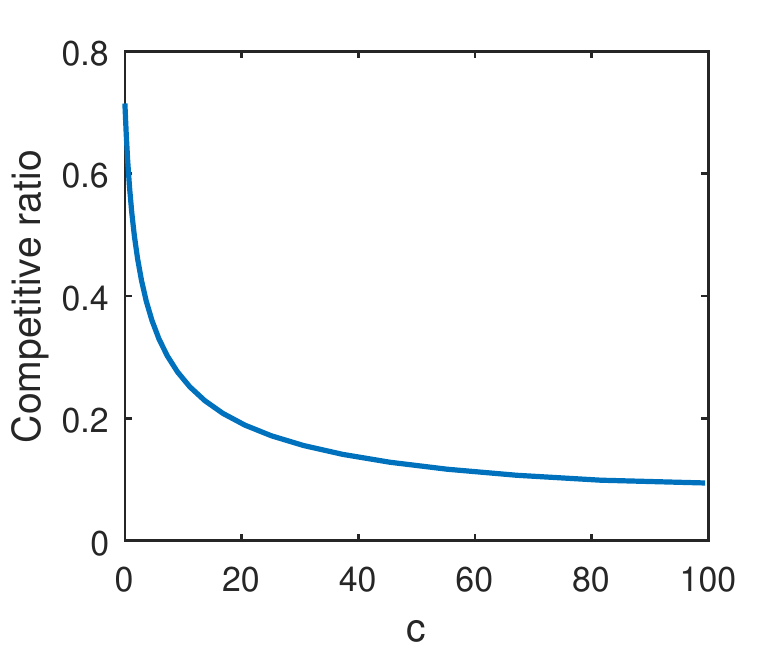}        
            \caption{}
              \label{fig-seq-log}     
 \end{subfigure}
 \caption{The competitive ratio achieved by the optimal smoothing for the sequential algorithm as a function of $c$ for $\psi(u) = \min(.75,\; u,\; .5 u+.25)$ (a) and $\psi(u) = \log(1+u)$ with $u_{\max} = 100$ (b). } \label{fig-seq}
\end{figure}

\section{Discussion and related work}
\label{discussion}

We discuss results and papers from two communities, computer science theory and machine learning, related to this work. 

\paragraph{Online convex optimization.}
In \cite{devanur2012online}, the authors proposed an optimal algorithm for adwords with differentiable concave returns (see examples in section \ref{CA}). Here, ``optimal'' means that they construct an instance of the problem for which competitive ratio bound cannot be improved, hence showing the bound is tight. 
The algorithm is stated and analyzed for a twice differentiable, separable $\psi(u)$.
The assignment rule for primal variables in their proposed algorithm is explained as a continuous process. A closer look reveals that this algorithm 
falls in the framework of algorithm \ref{algorithm2}, with the only difference being that at each step, $(\tilde{x}_t, \tilde{y}_t)$ are chosen such that
\begin{nips}
  \[
\begin{array}{l}
   \tilde{x}_t \in \argmax \inner{x}{A_t^{T} \tilde{y}_t}\\
    \forall i\in [n]: \quad \tilde{y}_{t,i} = \nabla \psi_i(v_i(u_i)), \quad u_i = (\sum_{t=1}^{t} A_{s} \tilde{x}_s)_{i},
\end{array}
\]
\end{nips}
where $v_i : \mathbb{R}_{+} \mapsto \mathbb{R}_{+}$ is an increasing differentiable function given as a solution of a nonlinear differential equation that involves $\psi_i$ and may not necessarily have a closed form. The competitive ratio is also given based on the differential equation. They prove that this gives the optimal competitive ratio for the instances where $\psi_1 = \psi_2 = \ldots={\psi_n}$. 
 
Note that this is equivalent of setting ${\psi_S}_i(u_i) =\psi(v_i(u_i)))$. Since $v_i$ is nondecreasing ${\psi_S}_i$ is a concave function. On the other hand, given a concave function ${\psi_S}_i(\mathbb{R}_{+}) \subset \psi_{i}(\mathbb{R}_{+}) $, we can set $v_i: \mathbb{R}_{+} \mapsto \mathbb{R}_{+}$ as $ v_i (u) = \inf\{z \; | \; \psi_i(z) \geq {\psi_S}_i(u)\}$.
Our formulation in section \ref{optimal-smoothing} provides a \emph{constructive} way of finding the optimal smoothing. It also applies to non-smooth $\psi$.

Recently, authors in \cite{azar2014online,buchbinder2014online,chan2015online} have provided a primal-dual online algorithm for the dual problem \eqref{fenchel-dual} that corresponds to the non-monotone primal objective $\psi\left(\sum_{t=1}^{m} A_t x_t\right) = \sum_{t=1}^{m} c_t^T x_t + G\left(\sum_{t=1}^{m} B_t x_t\right)$.
The primal and dual updates in their algorithm are presented as a continuous update based on a differential equation. They assume that $G^*$ is differentiable and that $\nabla G^*$ is monotone on $\mathbb{R}_+^n$, i.e., If $y \geq y'$, then $\nabla G^*(y) \leq \nabla G^*(y')$. In contrast, our assumption written in terms of $G^*$ for a differentiable function will become: If $\nabla G^*(y) < \nabla G^*(y')$, then $y \geq y'$, which is not equivalent to the assumption in \cite{buchbinder2014online}. When $G$ is separable the two assumptions coincide and this algorithm is similar to algorithm \ref{algorithm2} applied to the smooth function $G_M$ whose conjugate is given by $G_{M}^*(y) = \frac{1}{\gamma}  \sum_{i=1}^{n} \int_{0}^{y_i} {G_i^*}'(z)  \log(1-z/\theta) \;dz $. This smoothing coincides with Nesterov smoothing in the case of LP.

\paragraph{Online learning.}
As mentioned before, the dual update in Algorithm \ref{algorithm1} is the same as in Follow-the-Regularized-Leader (FTRL) algorithm with $-\psi^*$ as the regularization. This primal dual perspective has been used in \cite{shalev2007primal} for design and analysis of online learning algorithms. In the online learning literature, the goal is to derive a bound on \emph{regret} that optimally depends on the horizon, $m$. The goal in the current paper is to provide competitive ratio for the algorithm that depends on the function $\psi$. Regret provides a bound on the duality gap, and in order to get a competitive ratio the regularization function should be crafted based on $\psi$. A general choice of regularization which yields an optimal regret bound in terms of $m$ is \emph{not} enough for a competitive ratio argument, therefore existing results in online learning do not address our aim. 

\subsubsection*{Acknowledgments}
The authors would like to thank James Saunderson, Ting Kei Pong, Palma London, and Amin Jalali for their helpful comments and discussions.

\subsubsection*{References}
\bibliographystyle{plain}	
\bibliography{OnRef,Online}

\begin{full}
\appendix
\section*{Appendix}

\renewcommand{\thesubsection}{\Alph{subsection}}

\subsection{Proofs}\label{proofs}
\textbf{Proof of Lemma \ref{duality-gap}:} Using the definition of $D_{\rm sim}$, we can write: 

\begin{align*}
D_{\rm sim} &=  \sum_{t=1}^{m} \sigma_t(A_t^T \tilde{y}_t) - \psi^{*}(\tilde{y}_{m+1})\\
&=  \sum_{t=1}^{m} \inner{ A_t \tilde{x}_t}{\tilde{y}_t} - \psi^{*}(\tilde{y}_{m+1})\\
&\leq  \sum_{t=1}^{m} (\psi(\sum_{s=1}^{t}A_s \tilde{x}_s) - \psi(\sum_{s=1}^{t-1}A_s \tilde{x}_s) ) - \psi^{*}(\tilde{y}_{m+1})\\
&= \psi(\sum_{s=1}^{m}A_s \tilde{x}_s) -\psi(0) - \psi^{*}(\tilde{y}_{m+1}) ,
\end{align*}

\noindent where in the inequality follows from concavity of $\psi$, and the last line results from the sum telescoping. Similarly, we can bound $D_{\rm seq}$:
\begin{align}
D_{\rm seq} &=  \sum_{t=1}^{m} \sigma_t(A_t^T \hat{y}_t) - \psi^{*}(\hat{y}_{m+1})\\ \notag
& = \sum_{t=1}^{m} \inner{A_t \hat{x}_t}{\hat{y}_t} - \psi^*(\hat{y}_{m+1})\\ \notag
&=  \sum_{t=1}^{m} \inner{ A_t \hat{x}_t}{ -\hat{y}_{t+1}+\hat{y}_t} +\sum_{t=1}^{m} \inner{ A_t \hat{x}_t}{ \hat{y}_{t+1}}- \psi^{*}(\hat{y}_{m+1})\\ \notag
&\leq \sum_{t=1}^{m} \inner{ A_t \hat{x}_t}{ -\hat{y}_{t+1}+\hat{y}_t}  + \sum_{t=1}^{m} (\psi(\sum_{s=1}^{t}A_s \hat{x}_s) - \psi(\sum_{s=1}^{t-1}A_s \hat{x}_s) ) - \psi^{*}(\hat{y}_{m+1})\\ \notag
&= \sum_{t=1}^{m} \inner{ A_t \hat{x}_t}{ -\hat{y}_{t+1}+\hat{y}_t}  +\psi(\sum_{s=1}^{m}A_s \hat{x}_s) - \psi(0)- \psi^{*}(\hat{y}_{m+1}) .
\end{align}

To provide intuition about the above inequalities we have plotted the derivative of a concave function $\psi$ defined on $\mathbf{R}^+$ in Figure \ref{Riemann}. The quantity $ \sum_{t=1}^{m} \inner{ A_t \tilde{x}_t}{\tilde{y}_t}$ is a right Riemann sum approximation of the integral of $\psi'$ and lower bounds the integral (Figure \ref{right-Riemann}). The quantity $ \sum_{t=1}^{m} \inner{ A_t \hat{x}_t}{\hat{y}_t}$ is a left Riemann sum approximation of the integral and upper bounds the integral (Figure \ref{left-Riemann}). The area of hatched rectangle bounds the error of the left Riemann sum and is equal to $\inner{ A_t \hat{x}_t}{ -\hat{y}_{t+1}+\hat{y}_t}$.

\def\ta{1.87}
\def\tb{1.13}
\def\tc{.39}
\def\ra{2.7187}
\def\rb{1.62}
\def\rc{.95}
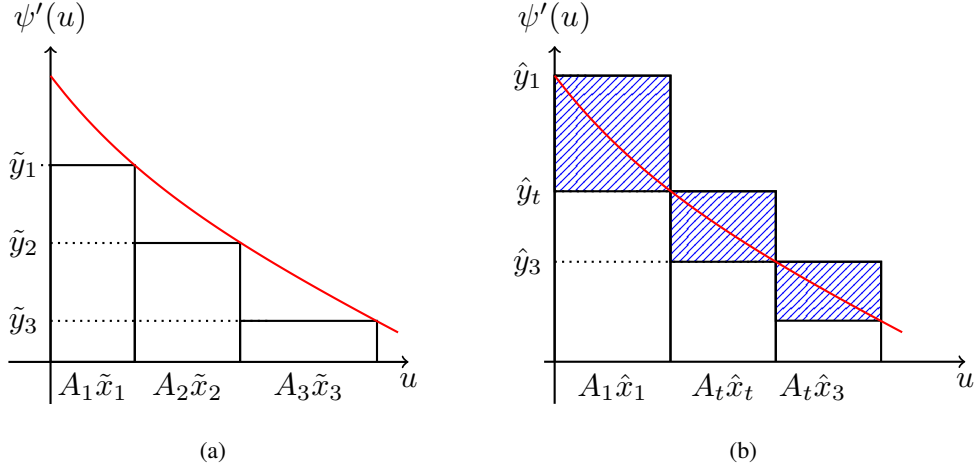
\begin{figure}
\begin{subfigure}{.5\textwidth}
\begin{center}
\begin{tikzpicture}[thick,scale=1.4, every node/.style={scale=1.2}]
  \draw[->] (-0.2,0) -- (3.6,0)node[below] {$u$};
  \draw[->] (0.2,-.4) -- (0.2,3) node[above] {$\psi'(u)$};
  \draw (0.2,0) rectangle (1,\ta) ;
  \draw (.6,0) node[below]{$A_{1} \tilde{x}_{1}$};
  \draw (.2,\ta) node[left]{$\tilde{y}_{1}$};
  \draw[dotted] (0.2,\ta) -- (.1,\ta);
  \draw (1,0) -- (1,\tb) -- (2,\tb) -- (2,0);
  \draw (1.5,0) node[below]{$A_{2} \tilde{x}_{2}$};
  \draw (.2,\tb) node[left]{$\tilde{y}_{2}$};
  \draw[dotted] (0.2,\tb) -- (1,\tb);
  \draw (2,0) -- (2,\tc) -- (3.3,\tc) -- (3.3,0);
  \draw (2.65,0) node[below]{$A_{3} \tilde{x}_{3}$};
  \draw (.2,\tc) node[left]{$\tilde{y}_{3}$};
  \draw[dotted] (0.2,\tc) -- (2.3,\tc);
  \draw[scale=1,domain=0.2:3.5,smooth,variable=\x,red] plot ({\x},{exp(-\x) -.5*\x+2});
\end{tikzpicture}
\end{center}
\caption{}
\label{right-Riemann}
\end{subfigure}
\begin{subfigure}{.5\textwidth}
\begin{center}
\begin{tikzpicture}[thick,scale=1.4, every node/.style={scale=1.2}]
  \draw[->] (.1,0) -- (4.1,0)node[below] {$u$};
  \draw[->] (0.2,-.4) -- (0.2,3) node[above] {$\psi'(u)$};
  \draw (.2,0) -- (.2,\ra) -- (1.3,\ra) -- (1.3,0);
  \draw (.75,0) node[below]{$A_{1} \hat{x}_{1}$};
  \draw (.2,\ra) node[left]{$\hat{y}_{1}$};
  \draw[dotted] (.2,\ra) -- (.5,\ra);
  \draw (1.3,0) -- (1.3,\rb) -- (2.3,\rb) -- (2.3,0);
  \draw (1.8,0) node[below]{$A_{t} \hat{x}_{t}$};
  \draw (0.2,\rb) node[left]{$\hat{y}_{t}$};
  \draw[dotted] (0.2,\rb) -- (1.5,\rb);
  \draw (2.3,0) -- (2.3,\rc) -- (3.3,\rc) -- (3.3,0);
  \draw (2.65,0) node[below]{$A_{t} \hat{x}_{3}$};
  \draw (0.2,\rc) node[left]{$\hat{y}_{3}$};
  \draw[dotted] (0.2,\rc) -- (2.3,\rc);
  \draw[pattern = north east lines, pattern color = blue](.2,\rb) rectangle (1.3,\ra);
  \draw[pattern = north east lines, pattern color = blue](1.3,\rc) rectangle (2.3,\rb);
  \draw[pattern = north east lines, pattern color = blue](2.3,.39) rectangle (3.3,\rc);
  \draw[scale=1,domain=0.2:3.5,smooth,variable=\x,red] plot ({\x},{exp(-\x) -.5*\x+2});
\end{tikzpicture}
\end{center}
\caption{}
\label{left-Riemann}
\end{subfigure}
\caption{Interpretation of $\sum_{t=1}^{m} \sigma_t(A_t^T \tilde{y}_t)$ and $\sum_{t=1}^{m} \sigma_t(A_t^T \hat{y}_t)$ as Riemann sums.}
\label{Riemann}
\end{figure}
When $\psi$ is differentiable with Lipschitz gradient, we can use the following inequality that is equivalent to Lipschitz continuity of the gradient:

\begin{align*}
\psi(u') \geq \psi(u) + \inner{\nabla \psi(u)}{u'-u} - \frac{1}{2\mu} \norm{u-u'}^2   \quad u,u' \in K
\end{align*}
\noindent (see, for example, \cite[section 2.1.1]{nesterov2004introductory}) to get
\begin{align}	D_{\rm seq}& =  \sum_{t=1}^{m} \sigma_t(A_t^T \hat{y}_t) - \psi^{*}(\hat{y}_{m+1})\\ \notag
& = \sum_{t=1}^{m} \inner{A_t \hat{x}_t}{\hat{y}_t} - \psi^*(\hat{y}_{m+1})\\ \notag
	&\leq \sum_{t=1}^{m} \frac{1}{2\mu} \norm{A_t \hat{x}_t}^2  + \sum_{t=1}^{m} (\psi(\sum_{s=1}^{t}A_s \hat{x}_s) - \psi(\sum_{s=1}^{t-1}A_s \hat{x}_s) ) - \psi^{*}(\hat{y}_{m+1})\\ \notag
	&= \sum_{t=1}^{m} \frac{1}{2\mu} \norm{A_t \hat{x}_t}^2 +\psi(\sum_{s=1}^{m}A_s \hat{x}_s) - \psi(0)- \psi^{*}(\hat{y}_{m+1}) .
\end{align}
\qed

\textbf{Proof of Lemma 2:}
Let $(y,\beta)$ be a feasible solution for problem \eqref{cont-problem}. Note that $y \geq 0$ since $\dom \psi^* \subset \mathbb{R}_{+}$ by the fact that $\psi$ is non-decreasing. Let $\bar{y}(u) = \inf_{s \leq u} y(s)$. Note that $\bar{y}$ is continuous.
Define 
\begin{align*} \beta(u) = \frac{\int_{s=0}^{u} y(s)\; ds  - \psi^*(y(u))}{\psi(u)}, \qquad  \bar{\beta}(u) = \frac{\int_{s=0}^{u} \bar{y}(s)\; ds  - \psi^*(\bar{y}(u))}{\psi(u)},\end{align*}
with the definition modified with the right limit at $u =0$. For any $u$ such that $\bar{y}(u) = y(u)$, we have:
\begin{align*}
\beta(u) = \frac{\int_{s=0}^{u} y(s)\; ds - \psi^*(y(u))}{\psi(u)} \geq \frac{\int_{s=0}^{u} \bar{y}(s) \; ds- \psi^*(\bar{y}(u))}{\psi(u)}  = \bar{\beta}(u).
\end{align*}

Now, we consider the set $\{u \; | \; \bar{y}(u) \neq y(u)\}$. By the definition of $\bar{y}$, we have $\bar{y}(0) = y(0)$. Since both functions are continuous, the set $\{u \; | \; \bar{y}(u) \neq y(u)\}$ is an open subset of $\mathbb{R}$ and hence can be written as a countable union of disjoint open intervals. Specifically, we can define the end points of the intervals as:
\begin{align*}
      &a_0 = b_0 = 0,\\
      &a_{i} = \inf\{u > b_{i-1}\;|\; y(u)> \bar{y}(u)\}, \quad \forall i \in\{1,2,\ldots\}\\
	&b_{i} = \inf\{u > a_{i} \;|\; y(u) = \bar{y}(u)\}, \quad \forall i \in\{1,2,\ldots\}.
\end{align*}
 then $\{u \; | \; \bar{y}(u) \neq y(u)\} = \bigcup_{i \in\{1,2,\ldots\}} (a_i,\; b_i).$(See Figure \ref{alpha-figure-ybar})

\begin{figure}
\begin{center}
\begin{tikzpicture}[thick,scale=1.4, every node/.style={scale=1.2}]
  \draw[line width=1.5 pt,->] (-0.6,.5) -- (4,0.5)node[below] {$u$};
  \draw[line width=1.5 pt,->] (-.4,0) -- (-.4,3) node[above] {};
  \draw[line width=.75 pt,dotted] (2,0.5) -- (2,2);
  \draw[line width=.75 pt,dotted] (0,0.5) -- (0,2);
  \draw[line width=.75 pt,dotted] (2.5,0.5) -- (2.5,1.5);
  \draw[line width=.75 pt,dotted] (3.5,0.5) -- (3.5,1.5);
  
  \draw[line width=1.5 pt,scale=1,domain=-.4:0,smooth,variable=\x,blue] plot ({\x},{2});
 \draw[line width=1.5 pt,scale=1,domain=-0.4:0,dashed,thick,variable=\x,red] plot ({\x},{2});
  \draw[line width=1.5 pt,scale=1,domain=0:1,smooth,variable=\x,blue] plot ({\x},{\x +2});
 \draw[line width=1.5 pt,scale=1,domain=0:2,dashed,thick,variable=\x,red] plot ({\x},{2});
  \draw[line width=1.5 pt,scale=1,domain=1:2.5,smooth,variable=\x,blue] plot ({\x},{4-\x });
  \draw[line width=1.5 pt,scale=1,domain=2:2.5,dashed,thick,variable=\x,red] plot ({\x},{4-\x });
  \draw[line width=1.5 pt,scale=1,domain=2.5:3,smooth,variable=\x,blue] plot ({\x},{\x -1});
  \draw[line width=1.5 pt,scale=1,domain=2.5:3.5,dashed,thick,variable=\x,red] plot ({\x},{1.5 });
 \draw[line width=1.5 pt,scale=1,domain=3:3.75,smooth,variable=\x,blue] plot ({\x},{-\x +5});
 \draw (0,.5) node[below]{$a_1$};
 \draw (2,.5) node[below]{$b_1$};
 \draw (2.5,.5) node[below]{$a_2$};
 \draw (3.5,.5) node[below]{$b_2$};
\end{tikzpicture}
\end{center}
\caption{ An example of $y(u)$ (solid blue) and $\bar{y}(u)$ (dashed red).}\label{alpha-figure-ybar}
\end{figure}
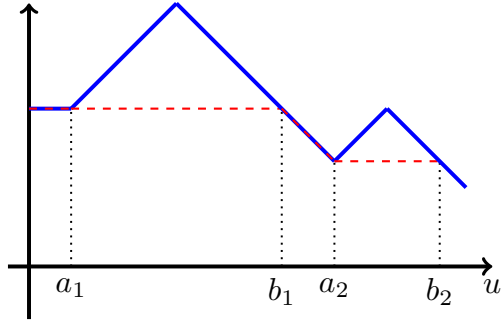 
For any $i \in \{1,2,\ldots\}$, we show that $\beta(u) \geq \bar{\beta}(u)$ on $(a_i,\; b_i)$. If $a_i=\infty$, then $b_i = \infty$, so we assume that $a_i < \infty$. By the definition of $a_i$ and $b_i$, $\bar{y}(u)$ is constant on $(a_i, \; b_i)$. Also, we have $y(a_i) = \bar{y}(a_i)$. Similarly, we have $y(b_i) = \bar{y}(b_i)$ whenever $b_i < \infty$. 

Since $\bar{y}(u) \leq y(u)$ for all $u$ and $y(a_i) = \bar{y}(a_i)$, we have
\begin{align}\label{beta1}
\beta(a_i) \geq \bar{\beta}(a_i).
\end{align}
If $b_i < \infty$, similarly by the fact that $y(b_i) = \bar{y}(a_i) = \bar{y}(b_i)$, we have
\begin{align}\label{beta2}
\beta(b_i) \geq \bar{\beta}(b_i).
\end{align}

Now we consider the case where $b_i = \infty$. In this case we have $\bar{y}(u) = \bar{y}(a_i)$ on $(a_i,\; \infty)$. We consider two cases based on the asymptotic behavior of $\psi$. If $\lim_{u \to \infty} \psi(u) = +\infty$ ($\psi$ is unbounded), then we have
\begin{align}\label{beta3}
\limsup_{u \to \infty}\beta(u) = \limsup_{u \to \infty} \frac{\int_{s=0}^{u} y(s)\; ds}{\psi(u)} \geq \limsup_{u \to \infty} \frac{\int_{s=0}^{u} \bar{y}(a_i)\; ds}{\psi(u)} = \lim_{u \to \infty}\bar{\beta}(u) .
\end{align}
Here we used the fact that $-\psi^*(y(u))$ is bounded. This follows from the fact $\psi^*$ is monotone thus:
\begin{align*}
-\psi^*(y(u)) \leq  -\psi^*(\bar{y}(a_i)),
\end{align*}
and $-\psi^*(\bar{y}(a_i)) < \infty$ because if $-\psi^*(\bar{y}(a_i)) = \infty$, then $\beta({a_i}) \geq \bar{\beta}({a_i}) = \infty$ which contradicts the feasibility of $(y,\beta)$. 

Now consider the case when $\lim_{u \to \infty} \psi(u) = M$ for some positive constant $M$. In this case, $- \psi^* \leq M$. We claim that $y(a_i) = 0$ and $\liminf_{t \to \infty} y(u) = 0$. Suppose $\liminf_{u \to \infty} y(u) > 0$, then $\limsup_{u \to \infty} \beta(u) = \infty$  since the numerator in the definition of $\beta$ tends to infinity while the denominator is bounded. But this contradicts feasibility of $(y,\beta)$. On the other hand, by the definition of $a_i$ and $b_i$ we should have $y(a_i) = \bar{y}(a_i) \leq \liminf_{u \to \infty} y(u)$. Combining this with the fact that $\bar{y}(a_i) \in \dom \psi^* \subset \mathbb{R}_{+}$, we conclude that $y(a_i) = 0$. Using that $y(a_i) = 0$ and $\liminf_{u \to \infty} y(u) = 0$, we get:
\begin{align}\notag
\limsup_{u \to \infty}\beta(u) &= \limsup_{u \to \infty} \frac{\int_{s=0}^{u} y(s)\; ds - \psi^*(y(u))}{\psi(u)} \\\notag
&\geq \lim_{u \to \infty} \frac{\int_{s=0}^{u} y(s)\; ds - \psi^*(0)}{M} \\\label{beta4}
&\geq \frac{\int_{s=0}^{a_i} \bar{y}(s)\; ds - \psi^*(0)}{M} = \lim_{u \to \infty}\bar{\beta}(u) ,
\end{align}

where in the last inequality we used the fact that $\bar{y}(u) = 0$ for $u \geq a_i$. 

Let $\psi'$ be the right derivative of $\psi$. Since $\psi$ is concave, $\psi'$ is non-increasing. Therefore, the interval $(a_i,b_i)$ can be written as $(a_i,u'] \cup [u',b_i)$ such that $\psi'(u) \geq \bar{y}(a_i)$ on $(a_i,u']$ and $\psi'(u) \leq \bar{y}(a_i)$ on $[u',b_i)$. Since $\psi'(u) \geq \bar{y}(a_i)$ on $(a_i,u']$ we have:
\begin{align*}
\int_{a_i}^{u} \bar{y}(s) \; d s &= \int_{a_i}^{u} \bar{y}(a_i) \; d s\\ 
&\leq \int_{a_i}^{u} \psi'(s) \; d s = \psi(u) - \psi(a_i),
\end{align*}
for all $u \in (a_i,u']$. This yields:
\begin{align*}
\bar{\beta}(a_i) &= \frac{\int_{s=0}^{a_i} \bar{y}(s)\; ds  - \psi^*(\bar{y}(a_i))}{\psi(u)}\\
&\geq \frac{\int_{s=0}^{a_i} \bar{y}(s)\; ds + \int_{a_i}^{u} \bar{y}(s) \; ds  - \psi^*(\bar{y}(a_i))}{\psi(a_i) + \psi(u) - \psi(a_i)}= \bar{\beta}(u).
\end{align*}
for all $u \in (a_i,u']$. Here we used the fact that if $c_1 \geq c_2 > 0$ and $d_2 \geq d_1 \geq 0$, then $$\frac{c_1}{c_2} \geq \frac{c_1 + d_1}{c_2+d_2}.$$ Similarly, we have $\bar{\beta}(b_i) \geq \bar{\beta}(u)$ for any $u \in  [u',b_i)$. Combining this with \eqref{beta1},\eqref{beta2},\eqref{beta3}, and \eqref{beta4}, we get:
\begin{align*}
\sup_{a_i \leq u \leq b_i} \bar{\beta}(u) &= \max(\bar{\beta}(a_i), \bar{\beta}(b_i))\\
 &\leq \max(\beta(a_i),\beta(b_i)) \leq \sup_{a_i \leq u \leq b_i} {\beta}(u).
\end{align*}

We conclude that $\bar{\beta}(u) \leq \beta(u)$ for all $t \geq 0$ hence $(\bar{y}, \beta)$ is a feasible solution for the problem.

%
%

\textbf{Proof of Theorem 2:}
Let $(y,\beta)$ be a feasible solution for problem \eqref{cont-problem}. By Lemma \ref{non-increasing-solution}, we can assume that $y$ is non-increasing. First, note that $y \geq 0$ since $\dom \psi^* = [0,\; \infty)$. Define $\bar{y}(u) = y(u)$ for $u \leq u'$ and $\bar{y}(u) = 0$ for $u > u'$.  We show that $(\bar{y},\beta)$ is also a feasible solution for \eqref{cont-problem} modulo the continuity condition. Define 
\begin{align*} \beta(u) = \frac{\int_{s=0}^{u} y(s)\; ds  - \psi^*(y(u))}{\psi(u)}, \qquad  \bar{\beta}(u) = \frac{\int_{s=0}^{u} \bar{y}(s)\; ds  - \psi^*(\bar{y}(u))}{\psi(u)}.\end{align*}

By the definition of $\bar{y}$, for all $u$, we have:
\begin{align}\label{y-bar-y}
 \int_{0}^{u} y(s)\; d s \geq \int_{0}^{u} \bar{y}(s)\; d s ,
\end{align}
and $\beta(u) = \bar{\beta}(u)$ for $u \in [0,\;u']$. 
Since $y(u)$ is non-increasing and $y(u) \geq 0$, $\lim_{u \rightarrow \infty} y(u)$ exists. We claim that $\lim_{u\rightarrow \infty} y(u) = 0$. To see this note that if $\lim_{u\rightarrow \infty} y(u) > 0$, then 
\begin{align*}
\lim_{u \to \infty}\int_{s=0}^{u} y(s)\; d s = \infty,
\end{align*}
which contradicts the fact that $\beta(u) \leq \beta$ for all $u$. For all $u \geq u'$, now we have:
\begin{align*}
\sup_{u \geq u'} {\beta}(u) \geq \lim_{u \to \infty} \beta(u) &= \frac{\lim_{u \to \infty} \int_{s=0}^{u} y(s)\; ds  - \psi^*(0)}{\psi(u')} \\
&\geq  \frac{\lim_{u \to \infty} \int_{s=0}^{u} \bar{y}(s)\; ds  - \psi^*(0)}{\psi(u')} = \bar{\beta}(u') ,
\end{align*}
where the equality follows from the fact that $\lim_{u\rightarrow \infty} y(u) = 0$, and in the last inequality, we used \eqref{y-bar-y}. Since $\bar{y}(u) = 0$ for $u > u'$, $\beta(u)$ is constant on $[u'\, \; \infty)$. Therefore, $\sup_{u \geq u'} \bar{\beta}(u) = \bar{\beta}(u')$. Combining this with the previous inequality we get:
\begin{align*}
\sup_{ u\geq u'} {\beta}(u) &\geq  \sup_{u \geq u'} \bar{\beta}(u').
\end{align*}
%
%
%
%
Therefore, we conclude that $\bar{\beta}(u) \leq \beta$ for all $u$. Thus $(\bar{y}, \beta)$ is also a feasible solution for \eqref{cont-problem} modulo the continuity condition. Note that $\bar{y}(u)$ may not be continuous at $u'$. However, we can find a sequence of continuous functions $z^{(j)}$ that converge pointwise to $y$ and $z^{(i)}(u) = 0$ for all $i$ and $u \geq u'$. To do so we consider a sequence of real number $\epsilon_i \to 0$. We define $z^{(i)}(u) = \bar{y}(u)$ for $u \in [0, \; u'-\epsilon_i) \cup [u', \; \infty)$. On $[u' - \epsilon_{i}, \; u']$ we define $z^{(i)}(u)$ to be a linear function that take values $y(u'-\epsilon)$ and $0$ on the endpoints. Define
\begin{align*} \beta_{z^{(i)}} =\sup_{u>0} \frac{\int_{s=0}^{u} z^{(i)}(s)\; ds  - \psi^*(z^{(i)}(u))}{\psi(u)}.\end{align*}
By upper semi-continuity of $\psi^*$, $\beta_{z^{(i)}}$ converges to $\bar{\beta}$. 

Let $\beta^*$ be the optimal solution for problem \eqref{cont-problem}. By the definition, there exits a feasible sequence $(y^{(j)}, \beta^{(j)})$ such that $\beta^{(j)}$ converges to $\beta^*$. Let $\bar{y}^{(j)}(u) = y^{(j)}(u)$ for $t \leq u'$ and $\bar{y}^{(j)}(u) = 0$ for $t > u'$. Note that $\bar{y}^{(j)}(u)$ may not be continuous at $u'$. However, we can find a sequence of continuous functions $(z^{(j i)}, \;\beta_{z^{(j i)}})$ as in above.  Now $\beta_{z^{(j j)}}$ converges to $\beta^*$.

\qed

\subsection{Distance from $l_p$ norm ball}\label{lin-lp} 
In this section we prove that the function:
\begin{align*}
G(u) = - d_1(u,\mathcal{B}_p)
\end{align*}
\noindent satisfies Assumption \ref{antitone} and find a lower bound on $\bar{\alpha}_{\psi}$ when $\psi$ is given by \ref{General-lp} with $G(u) = - l d_1(u,\mathcal{B}_{p})$.



For any $u \in \mathbb{R}_{+}^n$, there exists $\bar{u} \in \mathcal{B}_p$ such that $d_1(u,\mathcal{B}_p) = \norml{u-\bar{u}}$. the subdifferential of distance function is\footnote{For convex function we use $\partial$ to denote subdifferential.}:
\begin{align*}
\partial d_1(u,\mathcal{B}_p) = \partial \norml{u - \bar{u}} \cap N_{\mathcal{B}_p}(\bar{u}),
\end{align*}
\noindent where $N_{\mathcal{B}_p}(u) = \{\xi \;|\; \inner{\xi}{v - u } \geq 0, \; \forall v \in \mathcal{B}_p\}$ is the normal cone of $\mathcal{B}_p$ at $u$. In fact $d_1(u,\mathcal{B}_p) = \norml{u-\bar{u}}$ if and only if $\partial \norml{u - \bar{u}} \cap N_{\mathcal{B}_p}(\bar{u}) \neq \varnothing$. When $u \in {\rm int} \mathcal{B}_p$, $\bar{u} = u$ and $\partial d_1(u, \mathcal{B}_p) = \{0\}$. In order to find $\partial d_1(u,\mathcal{B}_p)$ when $u \notin {\rm int} \mathcal{B}_p$, we first find $\bar{u}$ in this case. For any $r \geq 0$, define $u \wedge r \in \mathbb{R}_{+}^n$ to be:
\begin{align*}
(u \wedge r)_i = \min(u_i,r)   \quad \forall i.
\end{align*}
Note that $\norm{u \wedge 0}_p = 0$  and $\norm{u \wedge (\max_{i} u_i)}_p = \norm{u}_p \geq 1$. Since $\norm{u \wedge r}_p$ is a continuous function of $r$, by the intermediate value theorem, there exists $r_u \in (0, \max_i u_i]$ such that $\norm{u \wedge r_u}_p = 1$. Now $\bar{u} = u \wedge r$. To see this note that:

 \begin{align}\label{sub1}
&\partial \norml{u -\bar{u}} \cap N_{\mathcal{B}_p}(\bar{u}) = \left\{\frac{1}{r_u^{p-1}} (u \wedge r_u)^{\circ (p-1)}\right\}  \quad \mbox{for} \quad r_u < \max_i u_i;\\ \label{sub2}
&\partial \norml{u -\bar{u}} \cap N_{\mathcal{B}_p}(\bar{u}) = \left\{\frac{z}{r_u^{p-1}}(u \wedge r_u)^{\circ (p-1)} \;|\; 0 \leq z \leq 1\right\} \quad \mbox{for} \quad r_u = \max_i u_i;
 \end{align} 
 
\noindent where $^{\circ (p-1)}$ denotes element-wise exponentiation. Now if $u' \leq u$, then $r_{u} \leq r_u'$ since $\norm{u \wedge r}_p \geq  \norm{u' \wedge r}_p$ for all $r$. Thus by \eqref{sub1} and \eqref{sub2}, there exists $y \in \partial d_1(u, \mathcal{B}_p)$ such that $y \geq \partial d_1(u',\mathcal{B}_p )$.

Now we can lower bound ${\alpha}_\psi$ when $\psi$ is given by \ref{General-lp} with $G(u) = - l d_1(u,\mathcal{B}_p)$. Let $(v,u) = (\sum_{t=1}^{m} c_t \tilde{x}_t, \sum_{t=1}^{m} B_t \tilde{x}_t)$. Note that by \eqref{sub1} when $y \in \partial \psi(u)$ with $u \notin \mathcal{B}_p$, then $\min_{i} y_i = -l$. Now by the definition of $l$ in \eqref{l-bound} and the explanation that followed it, we must have $ u \in \mathcal{B}_p$. 
If $u \in {\rm int} \mathcal{B}_p$, then $G$ is differentiable at $u$ and $\nabla G(u) = 0$ which yields $\alpha_{\psi}(v,u) = 0$. Now suppose $u \in {\rm bd}{\mathcal{B}_p}$. In this case we have:
\begin{align*}
\alpha_{\psi}(v,u) = \min_{y \in \partial G(u)}\frac{G^*(y)}{\psi\left(\sum_{t=1}^{m} A_t x_t\right)} &=  \min_{y \in \partial G(u)} \frac{-\norm{y}_{q}}{\sum_{t=1}^{n} c_t \bar{x}_t} \geq  \min_{y \in \partial G(u)} \frac{-\norm{y}_{q}}{\theta \mathbf{1}^T u}.
\end{align*}
Recall that $\theta =\min_{t} \min_{x \in F_t} \frac{c_t^T x}{\mathbf{1}^T {B_{t} x}}$. By \eqref{sub2}, we have:
\begin{align*}
\min_{y \in \partial G(u)} {-\norm{y}_{q}} =  - \frac{l}{(\max_i u_i)^{p-1}}.
\end{align*}
Therefore, 
\begin{align*}
\alpha_{\psi}(v,u) \geq  - \frac{l}{\theta}\frac{1}{(\mathbf{1}^T u )(\max_i u_i)^{p-1}}. \end{align*}
As $u$ varies on the ${\rm bd}{\mathcal{B}_p}$, the right hand side is lower bounded by $\frac{-l}{\theta}$. This yields $\alpha_\psi \geq -\frac{l}{\theta }$.

\subsection{Derivation of lower bounds on ${\alpha}_{\psi,\psi\Box\phi}$}\label{alpha-derivation}

We first derive a general inequality which will be specialized to different examples for bounding $\alpha_{\psi, \psi\Box\phi}$. Let $K = K_1 \times K_2$, with $K_1$ and $K_2$ two proper cones. Suppose $\psi(v,u) = H(v) + G(u)$, where $H: K_1 \mapsto \mathbb{R}$ is a non-decreasing, and $G:K_2 \mapsto \mathbb{R}$ is non-increasing and $l$ Lipschitz continuous. We assume $H(v) \geq \theta  u$ for all $(v,u) \in \sum_{t=1}^{m} A_t F_t$.  Note that $\psi^*(z,y) = H^*(z) + G^*(y)$. We set
\begin{align*}\phi^*(y)=\sum_{i=1}^{m}\frac{1}{\gamma} \left(\left(y_i -\frac{\theta}{(e-1)}\right) \log\left(1-\frac{ (e-1)}{\theta} y_i\right)-(1+\gamma)y_i\right),\end{align*}
where $\gamma = \log(1+ \frac{l (e-1)}{\theta})$. We let $\psi\Box\phi(v,u) = H(v) + G\Box\phi(u)$. Let $(v,u) \in \sum_{t=1}^{m} A_t \tilde{x}_t$. Since $\psi\Box\phi(0) = 0$, and the simultaneous algorithm does not decrease the objective,
\begin{align}\label{HGbox}
H(v) + G\Box\phi(u) \geq   \psi\Box\phi(0) = 0.
\end{align}

Let $(z,y) \in \partial \psi\Box\phi (u)$, then we have:
\begin{align}\notag
u_i &=  \nabla_i \phi^*(y)+ \tilde{\nabla}_i G^*(y) \\ \label{first-order-optimality} 
&= \frac{1}{\gamma} \log(1-\frac{(e-1)y_i}{\theta})-1 + \tilde{\nabla}_i G^*(y),
\end{align}
for some $\tilde{\nabla} G^*(y) \in \partial G^*(y)$. Using the previous identity, we can derive the following upper bound for $G\Box\phi(u)$:
\begin{align} \notag
G\Box\phi (u) &= \inner{y}{u} -G^*(y) -\phi^*(y)\\ \notag
 &=  \inner{y}{\tilde{\nabla} G^*(y) } -G^*(y)  + \inner{y}{\nabla \phi^*(y)} -\phi^*(y) \\ \notag
&= G(\tilde{\nabla} G^*(y) )+ \phi(\nabla \phi^*(y))\\ \notag
 &= G(\tilde{\nabla} G^*(y) )+ \frac{\theta}{(e-1)}\sum_{i=1}^{m}(u_i - \tilde{\nabla}_i G^*(y)+1) + \frac{1}{\gamma}\mathbf{1}^T y \\ \label{identity}
 &\leq G(\tilde{\nabla} G^*(y) )+ \frac{1}{(e-1)} H (v) + \frac{\theta}{(e-1)}\sum_{i=1}^{m}(1 - \tilde{\nabla}_i G^*(y))+ \frac{1}{\gamma}\mathbf{1}^T y .
\end{align}
%
Now we specialize the bound to the case where $G: \mathbf{R}_{+}^m \mapsto \mathbf{R}$ and $G(u) = - l \sum_{i=1}^{m}(u_i - 1)_{+}$. We lower bound $\alpha_{\psi,\psi\Box\phi}(v,u)$ when $u \leq \mathbf{1}$. In that case, \eqref{first-order-optimality} is satisfied with $\tilde{\nabla}_i G^*(y) = \mathbf{1}$. Thus from \eqref{identity} simplifies to:
\begin{align} \label{standalone-box}
 G\Box\phi (u) &\leq  \frac{1}{(e-1)}H (v) + \frac{1}{\gamma}\mathbf{1}^T y.
\end{align}

Combining this with \eqref{HGbox}, we get:
\begin{align} \label{D-bound-LP}
H(v)+G(u) = H(v) \geq  \frac{(1/e - 1)}{\gamma}\mathbf{1}^T y .
\end{align}

In the view of definition of $\alpha_{\psi,\psi\Box\phi}$, by using \eqref{standalone-box} and the fact that $G^*(y)=\mathbf{1}^T y$, we derive the following inequality:
\begin{align}\label{P-bound-LP}
H(v) + G\Box\phi(u) - G^*(y) &\leq ( 1+\frac{1}{e-1}) H(v)  + (\frac{1}{\gamma}-1) \mathbf{1}^T y .
\end{align}

Combining \eqref{D-bound-LP} and \eqref{P-bound-LP}, we get the following lower bound on $\alpha_{\psi,\psi\Box\phi}$:
\begin{align}\label{alphaHG}
{\alpha}_{\psi,\psi\Box\phi}(v,u) \geq 1 - (1+\frac{1}{e-1})\gamma + \bar{\alpha}_{H}.
\end{align}

\textbf{Online LP:} In this problem,
\begin{align}
\psi\left(\sum_{t=1}^{m} A_t x_t\right) = \sum_{t=1}^{m} c_t^T x_t + G\left(\sum_{t=1}^{m} B_t x_t\right),
\end{align}
with $G(u) = -l \sum_{i=1}^{n}(u_i-1)_{+}$. In this problem $H(v) = v$ is the identity function.

Let $(\tilde{v},\tilde{u}) = (\sum_{t=1}^{m} c_t^T \tilde{x}_t,\sum_{t=1}^{m} B_t \tilde{x}_t)$. If $\tilde{y}_{t,i} \leq -l$ for some $t$ and $i$, then by the definition of $l$ in \eqref{l-bound} and the explanation that followed it, we have $(B_t \tilde{x}_{t})_{i}=0$. On the other hand, we have $\nabla_{i} G \Box \phi(u) \leq - l$ when $\tilde{u}_i \geq 1$. Therefore, we conclude that $\tilde{u} \leq \mathbf{1}$.  Also, by the definition of $\theta$, we have $H(\tilde{v}) \geq \theta \mathbf{1}^T \tilde{u}$. Since $H$ is a linear function, $\alpha_H = 0$. Thus \eqref{alphaHG} yields
\begin{align}
{\alpha}_{\psi,\psi\Box\phi}(\tilde{v},\tilde{u}) \geq 1 - (1+\frac{1}{e-1})\gamma.
\end{align}

\textbf{Online graph formation and online Experiment design :} In this problem,  $H(U) = \log\det(U+A_0) - \log\det(A_0)$, $G(u) = - l (b - u)_{+}$ and $A_t x = (a_t^T a_t x, x)$. 

%
We can use the following identity for determinant of rank one update of a matrix 
$$\det(A+vv^T) = \det(A) (1+v^T A^{-1} v^T),$$
to derive a lower bound on $\log\det(U+A_0) - \log\det(A_0)$, 
 \begin{align*}\notag
 H(U) & =  \log\det(U+A_0) - \log\det(A_0)  \\ \notag
 &= \sum_{s=1}^{m} \log\det(A_0+ \sum_{t=1}^{s} a_t a_t^T \tilde{x}_t) -  \log\det(A_0+ \sum_{t=1}^{s-1} a_t a_t^T \tilde{x}_t)\\ \notag
 & = \sum_{t=1}^{m} \log(1+a_t^T \tilde{Y}_t a_t \tilde{x}_t)\\ \label{logbound-2}
 &\geq \log(1+\frac{1}{n}) \sum_{t=1}^{m} \tilde{x}_t = \theta u.
 \end{align*}

Let $(\tilde{U},\tilde{u}) = (\sum_{t=1}^{m} a_t a_t^T \tilde{x}_t, \sum_{t=1}^{m} \tilde{x}_t)$.
If $\tilde{y}_{t} \leq -l$ for some $t$, then by the definition of $l$ in \eqref{Kirchhoff-l}, $a_t^T \tilde{Y}_t a_t + \tilde{y}_t < 0 $ which results in $\tilde{x}_t = 0$. Given the fact that ${G \Box \phi}'(u)  = -l$ when $u \geq b$, we conclude that $\tilde{u} \leq b$. Therefore, by \eqref{alphaHG}, we get
\begin{align*}
\alpha_{\psi,\psi\Box\phi}(\tilde{U}, \tilde{u}) \geq  - (1+\frac{1}{e-1})\gamma.
\end{align*}

\end{full}

\end{document}

%% file: FullAbstract.tex
Online optimization covers problems such as online resource allocation, online bipartite matching, adwords (a central problem in e-commerce and advertising), and adwords with separable concave returns. 
We analyze the worst case competitive ratio of two primal-dual algorithms for a class of online convex (conic) optimization problems that contains the previous examples as special cases defined on the positive orthant. 
We derive a sufficient condition on the objective function that guarantees a constant worst case competitive ratio (greater than or equal to $\frac{1}{2}$) for monotone objective functions. 
We provide new examples of online problems on the positive orthant
and the positive semidefinite cone that satisfy the sufficient condition. We show how smoothing can improve the competitive ratio of these algorithms, and in particular for separable functions, we show that the optimal smoothing can be derived by solving a convex optimization problem. This result allows us to directly optimize the competitive ratio bound over a class of smoothing functions, and hence 
\emph{design} effective smoothing customized for a given cost function.